\documentclass[12pt]{article}

\usepackage[table]{xcolor} 
\usepackage{algorithm}
\usepackage{algorithmicx}
\usepackage{algpseudocode}
\usepackage{adjustbox}
\usepackage{amsthm}
\usepackage{amscd}
\usepackage{amsmath}
\usepackage{amssymb}
\usepackage{amsfonts}
\usepackage{bm}
\usepackage{bbm}
\usepackage{booktabs}
\usepackage{color, colortbl}
\usepackage{caption}
\usepackage{dsfont} % blackboard numerals
\usepackage{enumerate}
\usepackage{graphicx}
\usepackage{geometry}
\usepackage{hyperref}
\usepackage{mathabx}
\usepackage{multirow}
\usepackage{mathrsfs}
\usepackage{mdframed}
\usepackage{newtxtext}
\usepackage[round]{natbib}
\setcitestyle{aysep={}} 
\usepackage{setspace}
\usepackage{subcaption}
\usepackage{tikz-cd}
\usepackage{url}
\usepackage{wrapfig}
\usepackage{xcolor}

\newcommand{\R}{\mathbb{R}}
\newcommand{\E}{\mathbb{E}}

\newcommand{\s}{\mathbf{s}}
\newcommand{\x}{\mathbf{x}}

\newcommand{\blambda}{\boldsymbol{\lambda}}

\newcommand{\paren}[1]{\left( {#1} \right)}

\newcommand{\m} {\mathbf{m}}

\newcommand{\OR}{\mathrm{OR}}
\newcommand{\pkg}[1]{{\fontseries{b}\selectfont #1}}

\newtheoremstyle{mytheoremstyle}
  {\topsep} % Space above
  {0pt} % Space below
  {} % Body font
  {} % Indent amount
  {\bfseries} % Theorem head font
  {.} % Punctuation after theorem head
  {.5em} % Space after theorem head
  {} % Theorem head spec (can be left empty, meaning `normal')
\theoremstyle{mytheoremstyle}

\newtheorem{thm}{Theorem}

\newtheorem{lem}{Lemma}
\newtheorem{defn}{Definition}
\newtheorem{assum}{Assumption}

\allowdisplaybreaks

\newcommand{\blind}{0}

\title{Dependence-robust confidence intervals for capture-recapture surveys}

\if1\blind
{
  \author{[Blinded for review]
} \fi

\if0\blind
{
\author{Jinghao Sun\footnote{Jinghao Sun (jinghao.sun@yale.edu) is a PhD Candidate in Biostatistics at Yale School of Public Health, New Haven, CT, USA. [Corresponding author]},\\ Luk Van Baelen\footnote{Luk Van Baelen (luk.vanbaelen@sciensano.be) is a Senior Scientist at Department of Epidemiology and public health, Sciensano, Rue Juliette Wytsmanstraat, 14, Brussels 1050, Belgium.}, \\
Els Plettinckx\footnote{Els Plettinckx (els.plettinckx@sciensano.be) is a Principal Research Scientist at Department of Epidemiology and Public Health, Sciensano, Rue Juliette Wytsmanstraat, 14, Brussels 1050, Belgium. },\\ and Forrest W. Crawford\footnote{Forrest W. Crawford (forrest.crawford@yale.edu) is an Associate Professor of Biostatistics, Statistics \& Data Science, Operations, and Ecology \& Evolutionary Biology at Yale University, New Haven, CT, USA. } \\[1em]
} 
} \fi

\begin{document}

\maketitle

\begin{abstract}
  \noindent Capture-recapture (CRC) surveys are used to estimate the size of a population whose members cannot be enumerated directly.  CRC surveys have been used to estimate the number of Covid-19 infections, people who use drugs, sex workers, conflict casualties, and trafficking victims.  When $k$ capture samples are obtained, counts of unit captures in subsets of samples are represented naturally by a $2^k$ contingency table in which one element -- the number of individuals appearing in none of the samples -- remains unobserved.  In the absence of additional assumptions, the population size is not identifiable (i.e. point-identified). Stringent assumptions about the dependence between samples are often used to achieve point-identification. However, real-world CRC surveys often use convenience samples in which the assumed dependence cannot be guaranteed, and population size estimates under these assumptions may lack empirical credibility. In this work, we apply the theory of partial identification to show that weak assumptions or qualitative knowledge about the nature of dependence between samples can be used to characterize a non-trivial confidence set for the true population size.  We construct confidence sets under bounds on pairwise capture probabilities using two methods: test inversion bootstrap confidence intervals, and profile likelihood confidence intervals.  Simulation results demonstrate well-calibrated confidence sets for each method.  In an extensive real-world study, we apply the new methodology to the problem of using heterogeneous survey data to estimate the number of people who inject drugs in Brussels, Belgium. \\[1em]
  \end{abstract}

  \paragraph{Statement of Significance.} Capture-recapture surveys allow researchers to estimate the size of a population by measuring the overlap in at least two random samples from that population. This paper develops partial identification methodology to relax stringent dependence assumptions usually needed to obtain point identification in capture-recapture experiments.  Statistical dependence between samples can dramatically alter estimates of the size of the target population, but a fully parameterized model is not nonparametrically identifiable.  The purpose of this paper is to derive robust confidence intervals that can accommodate uncertainty in pairwise dependence between samples. The proposed method improves on traditional approaches, which must either assume certain dependence to be absent, or impose a prior distribution over dependence parameters. We have implemented open-source software for the proposed procedure in an R package for general CRC experiments.
  
  \section{Introduction}

  Estimating the size of a population is an important problem in demography, ecology, epidemiology, and public health research.  When the members of a population cannot be enumerated directly, probabilistic survey methods may be used to obtain statistical estimates of the population size.  Capture-recapture (CRC) surveys obtain several random samples from a population and record the number of unique individuals in each subset of samples. 
  Historically, CRC was first used in ecological studies, to monitor animal abundance and related demographic parameters \citep{seber1982estimation, williams2002analysis}. Recently, CRC surveys have been used in epidemiological studies to estimate the size of hidden or hard-to-reach populations, including undetected Covid-19 infections \citep{bohning2020estimating}, human trafficking and modern slavery \citep{silverman2020multiple}, men who have sex with men \citep{paz2011many}, sex workers \citep{kruse2003participatory}, people who inject drugs (PWID) \citep{hickman2009assessing}, methamphetamine users \citep{dombrowski2012estimating}, opiate users \citep{comiskey2001capture}, heroin users \citep{larson1994indirect}. CRC also has an important role in coverage evaluation studies for censuses and data integration/record linkage \citep{di2018population, aleshin2022multifile, manrique2022capture}.
  
  CRC analyses typically make four types of assumptions: 1) restrictions on \emph{inclusion dependence} between samples (e.g. when $k = 2$, it is assumed that the inclusion in one sample is independent of the inclusion in the other sample.) \citep{otis1978statistical, pollock1991review, agresti1994simple, chao2001overview}; 2) closed-population assumptions in which the population is assumed to be static in size and composition during the period of investigation, i.e., the effects of mortality, migration, and recruitment are negligible \citep{seber1982estimation}; 3) homogeneous capture probability, which means the probability of being captured in a certain sample is the same for each individual in the population; 4) distinguishability in captures, i.e. individuals are correctly identified between captures.

  Dependence assumptions in CRC studies are especially important because of identifiability issues. 
  Intuitively, a parameter is identifiable if it is theoretically possible to learn its true value after obtaining an infinite number of observations. 
  For non-hidden populations in survey studies with clear sampling frames, the (in)dependence structures may be known by design. However, for many CRC studies of hidden populations, investigators may not have precise prior knowledge about independence or dependence in samples. 
  In particular, when $k$ capture samples are obtained, counts of units captured in subsets of samples are represented naturally by a $2^k$ contingency table in which one element -- the number of individuals appearing in none of the samples -- remains unobserved \citep{fienberg1972multiple}.  Because the missing element can take any non-negative integer value,  the population size is not identifiable (i.e. point-identified) in the absence of additional inclusion dependence assumptions. The traditional CRC theory based on log-linear models \citep{bishop2007discrete, cormack1989log} is the most frequently used method for CRC in social sciences (e.g. \citep{hay2016estimating, xu2014estimating, kimber2008estimating, jones2016problem}). A full log-linear model has $2^k$ parameters, which is unidentified. To achieve point-identification, one or more parameters are usually assumed to be $0$, leading to specific \emph{inclusion dependence structures} among samples.

  It is widely known that misspecification of dependence in CRC samples may result in biased estimates of population size \citep[e.g.][]{tilling2001capture}.  Researchers have attempted to describe and deal with issues of dependence from different perspectives. Because a fully specified model with $2^k$ parameters is not identifiable, researchers have explored ways of modeling or assessing sensitivity to unknown dependence. \citet{hook2000accuracy} show that when $k>2$ surveys are available, ``internal validity analysis" can be conducted by comparing estimates under the full $k$ samples with those generated from all combinations of $k-1$ surveys. \citet{baffour2013investigation} investigate how the number of surveys used in CRC impacts bias in population size estimates. \citet{wolter1990capture}, \citet{bell1993using} and \citet{das2021doubly} take advantage of measured covariates in CRC surveys, or external information about population characteristics, e.g. estimated sex ratio from other demographic surveys. A simulation and sensitivity analysis approach has also been adopted \citep{brown1999methodological, brown2006dependence, gerritse2015sensitivity, aleshin2021revisiting}. 
  For example, in a saturated model with $2^k-1$ free parameters, \citet{gerritse2015sensitivity} fixes the $k$-way interaction (highest-order dependence) parameter at a given value so that all the rest of the model parameters are point-identified, and then, varies this chosen parameter to investigate its impact on the population size estimates. \citet{aleshin2021revisiting} describe a Bayesian approach for sensitivity analysis by imposing a prior distribution over unknown dependence parameters. 
  
  Where might additional information about dependence in samples come from? Often researchers have access to qualitative information about the \emph{pairwise} dependence structure of the target population, but rarely about higher-order dependence parameters. For example, two respondent-driven sampling (RDS) \citep{heckathorn1997respondent, crawford2018hidden, yauck2022population} samples from the same target population may start with similar sets of seeds, leading to positive dependence of inclusion. Likewise, administrative lists of individuals who interact with a medical clinic, social service provider, or law enforcement entity are sometimes used as samples in CRC studies.  But membership on these lists may not be independent: an individual who seeks medical care may be more likely to also seek social services, or be less likely to be arrested.  Alternatively, membership on a given list may preclude membership in another list. For example, clinics may serve non-overlapping groups of clientele, excluding patients from neighboring catchment areas, thereby inducing negative correlation in study capture indicators. When capture samples involve a sequence of in-person visits or interviews, subjects \citep[e.g. sex workers][]{kimani2013enumeration} included in the first visit may be more likely to be included in the second visit, due to familiarity and trust with interviewers. In addition, geographically disparate samples may be negatively dependent.  Because a fully parameterized model with $2^k$ unknowns is not identifiable, researchers must either make unverifiable assumptions, or use available auxiliary information about dependence. 
  
  In this paper, we provide a rigorous frequentist statistical framework for estimating population sizes when the dependence structures among samples can vary over a wide range, characterized by weak information about pairwise dependence.  In particular, we do not assume no $k$-way interactions among all the $k$ captures as those used in traditional hierarchical log-linear models. By specifying one or more bounds on odds ratios for pairwise sample inclusion in a CRC study, we show how to estimate intervals that contain the target population size with high probability, without the need to specify a prior distribution over dependence parameters. Investigators need only specify one or more bounds on pairwise sample dependence to use the proposed method. 
  Our approach uses ideas from the newly developed theory of \emph{partial identification} \citep{manski2003partial, tamer2010partial, molinari2020microeconometrics} to conduct statistical inference in two ways: by introducing test inversion bootstrap confidence intervals and profile likelihood confidence intervals. Here we focus on confidence intervals instead of point estimates because the population size itself is only partially identified, so no consistent point estimator exists without additional assumptions \citep{lewbel2019identification}. Because we make no distributional assumptions about unidentified dependence parameters, the approach is distinct from Bayesian methods that require the specification of a prior distribution over these parameters. 
  We have implemented an open-source R package (See Supplementary Appendix) for the proposed procedure with detailed documentation for general $k$-sample CRC experiments and pairwise restrictions of flexible forms and amounts.

  \section{Motivating application}
  
  This work is motivated by the empirical problem of estimating the number of people who inject drugs (PWID) in Brussels, Belgium. Knowing the size of this hidden population is vital to government and non-governmental organizations that provide services to PWID, including drug treatment services and harm reduction programs like syringe exchange.  We apply the new partial identification methodology using three samples of PWID collected by \citet{plettinckx2020estimates}: participants in an RDS fieldwork study, clients of a crisis intervention center and shelter, and participants at a low-threshold drug treatment center. A total of 306 unique individuals were sampled across three of the studies. The counts of individuals in each sample subset are shown in Figure \ref{fig:pwid1}.  
  
  In samples 2 and 3, information was obtained from registration systems, while sample 1 arises from an RDS study starting with a small number of seeds selected from low-threshold treatment centers or syringe exchange services \citep{van2020prevalence}. 
  The three samples are thus likely to exhibit positive pairwise dependence.  
  
  The remainder of this paper describes a statistical methodology for weak substantive knowledge about the nature of pairwise dependence between samples to compute confidence intervals for the hidden population size.  Using this new methodology, we present dependence-robust interval estimates for the number of PWID in Brussels in Section \ref{sec:app}.

  \begin{figure}%[h]
    \centering
    \includegraphics[width=0.5\textwidth]{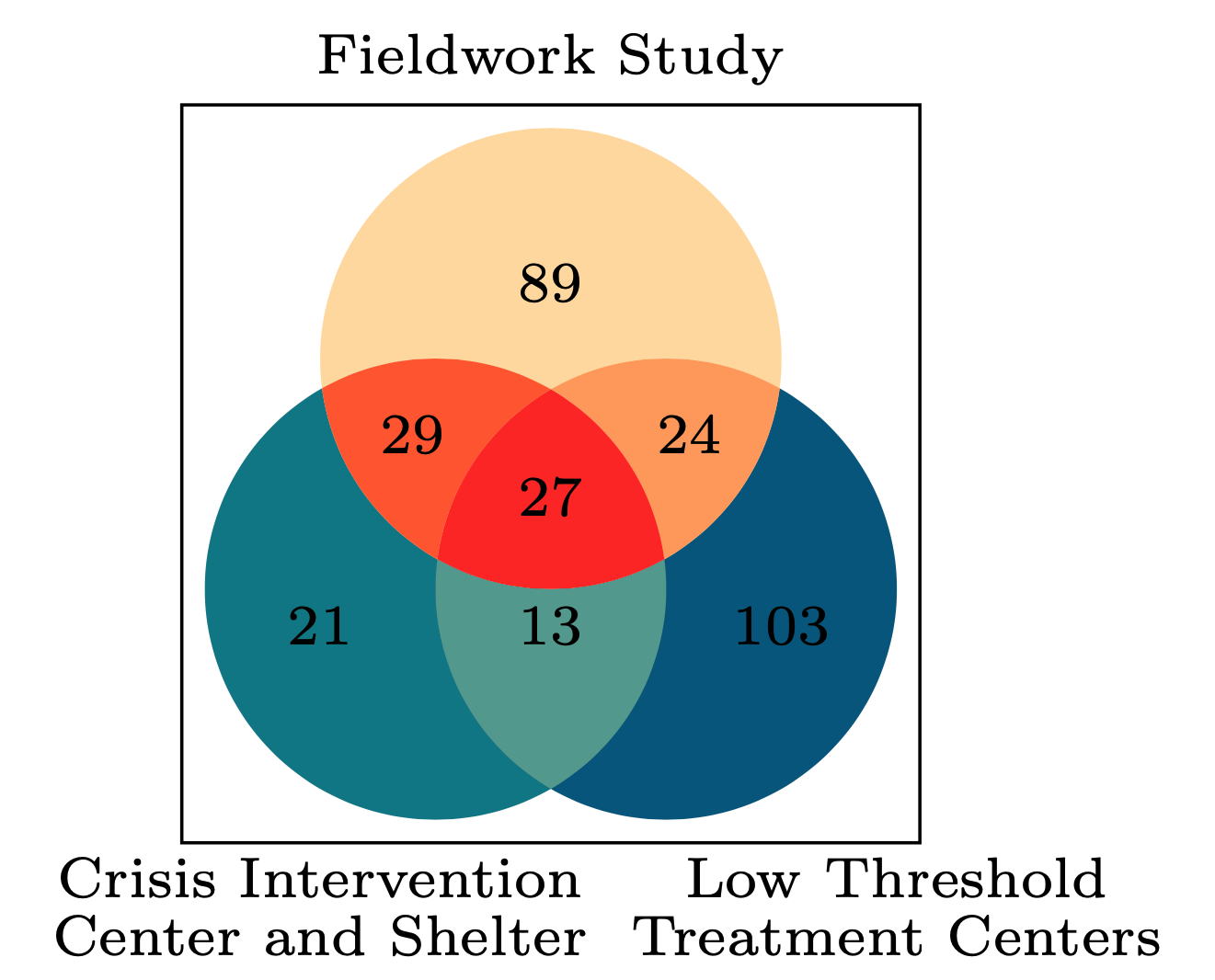}
    \caption{Illustration of data from three semi-overlapping samples of people who inject drugs in Brussels, Belgium \citep{plettinckx2020estimates}. Samples from low threshold drug treatment centers (``MSOC/MASS and Projet Lama") and a crisis intervention center and shelter (``Transit asbl") are obtained from registration systems. The sample from the fieldwork study was collected with RDS. A total of 306 unique PWID were sampled. }
    \label{fig:pwid1}
  \end{figure}

%%%%%%%%%%%%%%%%%%%%%%%%%%%%%%%%%%%%%

\section{Setting}

Consider a population consisting of $M$ distinguishable units. In this paper, we assume that the samples come from a closed population, capture probabilities are homogeneous within samples, and that sampled units are distinguishable so they can be matched between samples.
We obtain $k$ possibly dependent samples and record the number of units observed to fall within each of the $2^k$ subsets of samples.  Label these subsets $i=0,\ldots,c$, where $c=2^k-1$; the subset corresponding to label $i=0$ is the units not appearing in any sample, and the subset corresponding to $i=c$ is the units appearing in all $k$ samples.  Denote random variables $N_i$ as the number of units in the subset $i$, for $i = 0, \ldots, c$. Define the ``capture history'' of units in subset $i$ as a vector of sample indicators $\s_i \equiv (s_{i1}, \ldots, s_{ik})$. Define $\x_i \equiv  (x_{i0}, x_{i1}, \ldots, x_{ic})$ as a vector of $2^k$ elements mapping the capture history of subset $i$ to fixed effects in a model of $2^k$ parameters. 
Figure \ref{fig:illustration} shows an example of notation for $k=3$ capture samples, the values for $\s_i, \x_i$, and the incomplete contingency table representation.

\colorlet{myred}{red!40}
\definecolor{ione}{rgb}{0.00,0.38,0.53}
\definecolor{itwo}{rgb}{0.00,0.51,0.56}
\definecolor{ithree}{rgb}{0.36,0.64,0.59}
\definecolor{ifour}{rgb}{1.00,0.86,0.65}
\definecolor{ifive}{rgb}{1.00,0.63,0.37}
\definecolor{isix}{rgb}{1.00,0.36,0.16}
\definecolor{iseven}{rgb}{1.00,0.12,0.10}
    
\begin{figure}%[h]
\begin{minipage}{0.45\textwidth}
\centering
\includegraphics[width=0.7\textwidth]{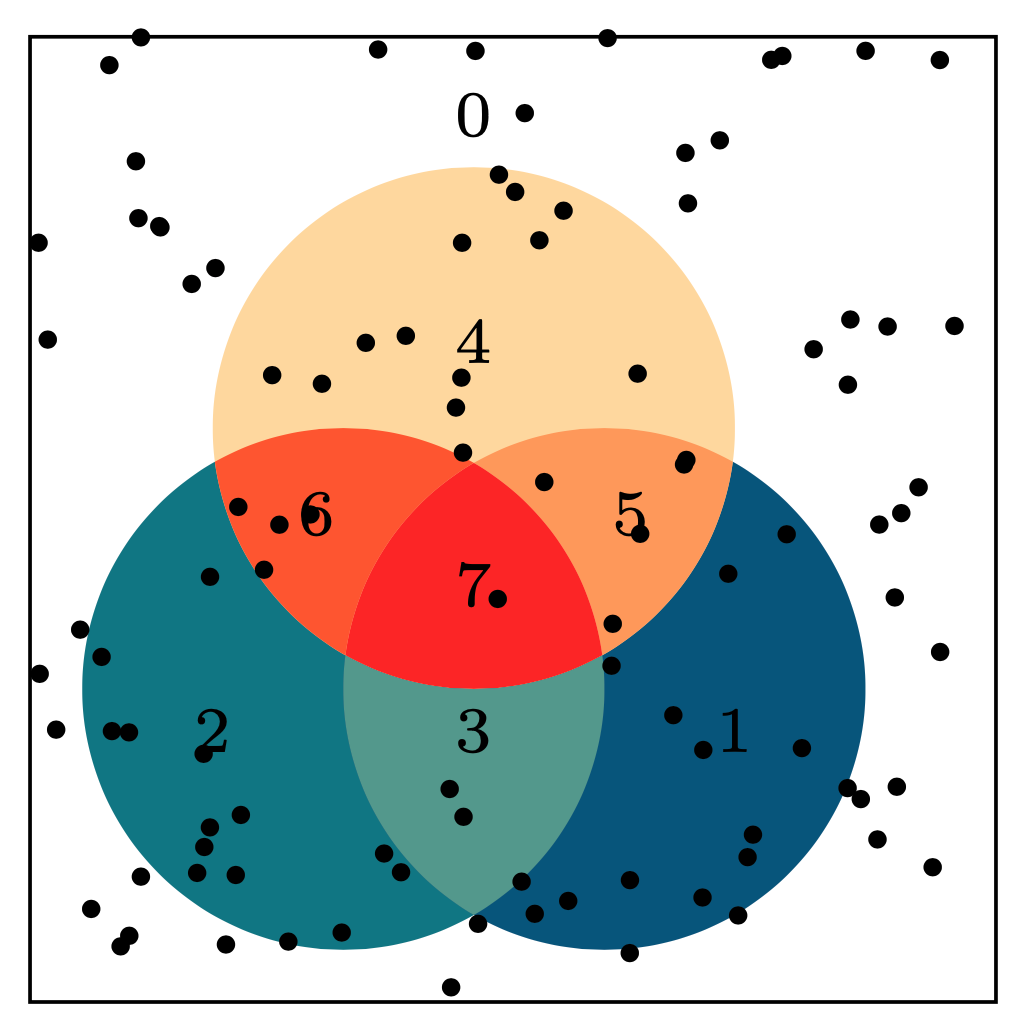}
\end{minipage}
\begin{minipage}{0.45\textwidth}
  \centering
    \begin{tabular}{cccc}   
    \toprule
        $i$ & $\s_i$ & $\x_i$ & $N_i$ \\
       \midrule
       0 & (0,0,0) & (1,0,0,0,0,0,0,0) & $N_0 = 49$ \\
       \rowcolor{ione}
       1 & (0,0,1) & (1,0,0,1,0,0,0,0) & $N_1 = 13$\\
       \rowcolor{itwo}
       2 & (0,1,0) & (1,0,1,0,0,0,0,0) & $N_2 = 14$\\
       \rowcolor{ithree}
       3 & (0,1,1) & (1,0,1,1,0,0,1,0) & $N_3 = 2$\\
       \rowcolor{ifour}
       4 & (1,0,0) & (1,1,0,0,0,0,0,0) & $N_4 = 13$\\
       \rowcolor{ifive}
       5 & (1,0,1) & (1,1,0,1,0,1,0,0) & $N_5 = 4$\\
       \rowcolor{isix}
       6 & (1,1,0) & (1,1,1,0,1,0,0,0) & $N_6 = 4$\\
       \rowcolor{iseven}
       7 & (1,1,1) & (1,1,1,1,1,1,1,1) & $N_7 = 1$\\
       \bottomrule \\
    \end{tabular}
  \end{minipage}

\begin{minipage}{1\textwidth}
\centering
\begin{tabular}{@{}lcccc@{}}
\toprule
\multirow{2}[3]{*}{Count} & \multicolumn{2}{c}{Not in sample 3} & \multicolumn{2}{c}{In sample 3} \\
\cmidrule(lr){2-3} \cmidrule(lr){4-5}
 & Not in sample 1 & In sample 1 & Not in sample 1 & In sample 1 \\
\midrule
Not in sample 2  & Missing & $N_4$ & $N_1$ & $N_5$ \\
In sample 2 & $N_2$ & $N_6$ & $N_3$ & $N_7$ \\
\bottomrule
\end{tabular}
\end{minipage}
\caption{Illustration of notation for a capture-recapture survey with $k=3$ samples.  At the upper left, a Venn diagram of samples shows disjoint subsets labeled by $i=0,\ldots,2^3-1$. Each dot represents an individual in the target population, for a total of $M=100$ units. At upper right, each subset $i$ corresponds to a $k$-vector of binary capture indicators $\s_i$, a $2^k$ design vector $\x_i$, and the count of individuals in subset $i$, $N_i$.  The table at the bottom shows the $2^k$ incomplete contingency table representation of CRC when $k = 3$.  The $N_0$ element, representing the units not captured in any of the three samples, is missing. }
\label{fig:illustration}
\end{figure}

For each subset $i$, let $m_i = \E[N_i]$ where expectation is defined with respect to the sampling design for the $k$ samples.  To describe the relationship between the $r$th and $t$th samples, where $r, t \in \{1,\ldots,k\}, r \neq t$, we define 
$ N_{d_1 d_2} (r, t) = \sum_{i= 1}^c N_i \mathbbm{1} \lbrace \boldsymbol{ s }_{ir} = d_1 \rbrace \mathbbm{1} \lbrace \boldsymbol{ s }_{it} = d_2  \rbrace, $ 
where $d_1, d_2 \in \{0, 1\}$. Then $N_{00} (r, t), N_{10} (r, t), N_{11} (r, t)$ are the counts of observed individuals who appear in neither sample $r$ nor sample $t$, in sample $r$ but not in sample $t$, and in both $r$ and $t$ respectively. Define the expected value 
$ m_{d_1 d_2} (r, t) \equiv \E\paren{N_{d_1 d_2} (r, t)} = \sum_{i= 1}^c m_i \mathbbm{1} \lbrace \boldsymbol{ s }_{ir} = d_1 \rbrace \mathbbm{1} \lbrace \boldsymbol{ s }_{it} = d_2  \rbrace, $
and $p_{d_1 d_2} (r, t) \equiv m_{d_1 d_2} (r, t)/M$.  Then the odds ratio (OR) for the capture probabilities in samples $r$ and $t$ is 
$ \OR_{rt} \equiv \frac{p_{11} (r, t)/p_{01} (r, t)}{p_{10} (r, t)/(p_{00} (r, t) + m_0/M)} = \frac{m_{11}(r,t)(m_{00}(r,t) + m_0)}{m_{10}(r,t)m_{01}(r,t)}. $
For example, when $k = 2$, $m_{00}(1,2) = 0$ and  
$ \OR_{1,2} = \frac{m_{11}(1,2)m_{0}}{m_{10}(1,2)m_{01}(1,2)}.$
We will assume throughout that $N_i \sim \text{Poisson}(m_i), i = 0, \ldots, c$ independently, and the population size is $M = \sum_{i=0}^c m_i$. This model is often called the ``Poisson model" \citep{cormack1979models, jolly1979unified}.

\section{Methods}

A CRC experiment with $k$ samples is a realization of $(N_1,\ldots,N_c)$. Suppose that $n$ identically and independently distributed contingency tables are available from repeated CRC studies for the same target population under the same sampling design.  Denote data in the $\ell$-th such table as $\mathbf{N}_\ell = (N_{\ell 1}, \ldots, N_{\ell c})$, for $\ell = 1,\ldots, n$. Define the average occupancy of the $i$th subset as $\bar{N}_{(n)i} = \frac{1}{n}\sum_{\ell = 1}^n N_{\ell i}$.

To construct frequentist confidence intervals, we employ an asymptotic regime in which the number of sampled units $n \rightarrow \infty$. 
In this ideal case, $M$ is still unidentified, even though we will have perfect knowledge of $\left( m_1, \ldots,  m_c \right)$, whose sum provides a lower bound for $M$. When qualitative/domain knowledge about CRC experiments is available, bounds on $M$ will be more informative. However, it may not be point-identified especially when the domain knowledge is inadequate. Making rigorous statistical inference for the population size which is possibly partially identified is the major methodological challenge in this work. In the following, we first formally define partial identification. Then, we develop statistical inference results under the restrictions on pairwise dependence using partial identification theory. 

\subsection{Partial identification}

Let $\mathbf{m} = (m_1, \ldots, m_c)$ be the mean occupancies of each CRC subset, and let $M= m_0 +\sum_{i=1}^c m_i$.  The model parameter vector is $\theta = (M, \mathbf{m}) \in \Theta$, where $\Theta \subseteq \mathbb{R}_+^{c+1}$ is the model parameter space.  Under the Poisson model, $\mathbf{N}_\ell\sim P \in \mathcal{P} = \lbrace P_{\theta}: \theta \in \Theta \rbrace$ i.i.d.  The nonempty model $\Theta$ is usually defined through identification assumptions that are formed by empirical knowledge, as we discuss below.  For a given $\theta\in\Theta$, the probability mass function of $\mathbf{N}_\ell$, as the sample criterion function, is 
$ p_{\theta}(\mathbf{N}_\ell) = \prod_{i = 1}^{c} \frac{e^{-m_i}m_i^{N_{li}}}{N_{li}!}. $ 
Under the true data generating process $P = P_{\theta^{\ast}}$, where $\theta^{\ast} = (M^{\ast}, \mathbf{m}^{\ast})$ is the true parameter vector, define the population criterion function $L: \Theta \rightarrow \mathbb{R}$ as 
$ L(\theta) = \sum_{i = 1}^c -m_i + m_i^{\ast}\log m_i ,$
which equals $\E_{P}\left[{\log p_{\theta}(\mathbf{N}_\ell)} \right]$ up to a constant. 

The identification set $\Theta_I(P)$ for $\theta$ is the set of maximizers of $L$,
$ \Theta_I(P) = \lbrace \theta \in \Theta: L(\theta) = \sup_{\nu  \in \Theta} L(\nu) \rbrace .$
Define the identification set $M_I(P)$ for the parameter of interest $M$ as 
$ M_I(P) = \lbrace M: (M, \m) \in \Theta_I(P) \text{ for some } \m \rbrace$,
which is the projection of $\Theta_I(P)$ on the axis of $M$.
When $M_I(P)$ contains only one element (i.e. $M_I = \lbrace M^{\ast} \rbrace$), the population size $M$ is \emph{point-identified}. When $\lbrace M^\ast\rbrace \subsetneq M_I \subsetneq \mathbb{R}_+$, the population size $M$ is \emph{partially identified}. Usually, when additional assumptions are imposed, the size of the identification set $M_I$ will shrink accordingly. Before we proceed, we make three regularity assumptions to ensure that the identification set is non-trivial and well-defined.   

\begin{assum}[Feasibility]
\label{assum:feasibility}
The parameter space $\Theta \subseteq \R_{+}^{c+1}$ is non-empty.
\end{assum}
\begin{assum}[Compactness]
\label{assum:compactness}
The parameter space $\Theta$ is closed and bounded away from 0 and $\infty$.
\end{assum}
\begin{assum}[Correctness]
\label{assum:correctness}
The true parameter $\theta^\ast \in \Theta$.
\end{assum}

Searching for the true model parameter vector in a null set is meaningless, so Assumption \ref{assum:feasibility} requires that the practitioners verify the non-emptiness of $\Theta$ after specifying it. 
Assumption \ref{assum:compactness} is realistic in empirical studies of large finite populations in which not every individual in the population is sampled. This is a technical condition that is primarily used in proofs. Assumption \ref{assum:correctness} requires that the true mean occupancy of each sample subset, as well as the true population size, is an element of the parameter space $\Theta$. 

Researchers often have qualitative knowledge about dependence among samples, and we express this knowledge in the form of bounds on the dependence between pairs of samples. In the following, we use the odds ratio (OR) between samples to quantify pairwise dependence. Note that with observable data, to get nontrivial bounds on the population size, researchers only need to know some but not necessarily all pairwise dependence relationships, and all the higher-order dependence among samples are left unspecified for robustness. Intuitively, the specification of one or more pairwise dependence relationships, along with observed elements of the $k$-way contingency table, impose shape constraints on the dependence structure between samples.  These constraints meaningfully constrain the set of possible population sizes, based on which we apply partial identification methodology to construct confidence intervals for the population size.

Suppose that we know the dependence between samples $r$ and $t$, and believe that 
\[ \OR_{rt} = \frac{m_{11}(r,t)(m_{00}(r,t) + m_0)}{m_{10}(r,t)m_{01}(r,t)} \in [\eta, \xi ]. \] 
Since $M = m_0 + \sum_{d_1, d_2} m_{d_1 d_2} (r, t)$, we obtain restrictions of the form 
\[ \eta \le \frac {m_{11}(r,t) \left[ M - m_{10}(r,t) - m_{01}(r,t) - m_{11}(r,t)\right]}{m_{10}(r,t)m_{01}(r,t)} \le \xi , \]
where $0\le\eta<\xi$. Suppose we have dependence restrictions on $\omega$ distinct pairs of samples, where $\omega\le \binom{k}{2}$, and denote these pairs as $\paren{r_j, t_j}$, and their restrictions as $\paren{\eta_{j}, \xi_{j}}$, $ j = 1, \ldots, \omega .$ Contradictory OR conditions are excluded by Assumption  1 (Feasibility). Violations of feasibility can be detected by checking whether there exist $m_{i}$'s that obey all the OR inequalities. In practice, when using our R package, warning and error messages would be generated if pairwise ORs become contradictory.

Note that $\omega$ only depends on the available domain knowledge of pairwise dependence and may only grow slowly as $k$ grows. Therefore, the proposed approach is scalable when the number of captures increases. 
Define $m_{d_1 d_2}^{(j)} \equiv m_{d_1 d_2}(r_j, t_j) $.
The following result describes the identification set for the population size $M^*$ under restrictions on pairwise dependence.

\begin{lem}[Identification set of $M^*$ under restrictions on pairwise dependence]
\label{lem:Theta_pair}
Given pairwise restrictions $\paren{r_j, t_j, \eta_{j}, \xi_{j}}_{j = 1}^{\omega}$, the model space $\mathcal{P}$ has corresponding parameter space
\begin{equation}
\label{eq:Theta_pair}
\begin{split}
\Theta =\Bigg\{ & (M, \m) \in \mathbb{R}_+^{c+1}: \eta_j \le \frac {m_{11}^{(j)} \left[ M - m_{10}^{(j)}  - m_{01}^{(j)}  - m_{11}^{(j)} \right]}{m_{10}^{(j)} m_{01}^{(j)} } \le \xi_j, \text{ for } j=1,\ldots,\omega   \Bigg\} .
\end{split}
\end{equation}
When the true parameter vector is $\theta^\ast = (\m^\ast, M^\ast)$, define $m_{d_1 d_2}^{\ast (j)} = \sum_{i= 1}^c m_i^{\ast} \mathbbm{1} \lbrace { s }_{ir_j} = d_1 \rbrace \mathbbm{1} \lbrace { s }_{it_j} = d_2  \rbrace $. Then the identification set for $M^*$ is 
\begin{equation}
\label{eq:MIP_pair}
\begin{split}
M_I(P) = \Bigg[ 
& \max_{j \in \{1, \ldots, \omega \}} \left\lbrace \eta_j \frac{m_{10}^{\ast (j)}m_{01}^{\ast (j)}}{m_{11}^{\ast (j)}} + m_{10}^{\ast (j)} + m_{01}^{\ast (j)} + m_{11}^{\ast (j)}
 \right\rbrace, \\
& \min_{j \in \{1, \ldots, \omega \}}  \left\lbrace \xi_j \frac{m_{10}^{\ast (j)}m_{01}^{\ast (j)}}{m_{11}^{\ast (j)}} + m_{10}^{\ast (j)} + m_{01}^{\ast (j)} + m_{11}^{\ast (j)} \right\rbrace
\Bigg] .
\end{split}
\end{equation}
\end{lem}

%%%%%%%%%%%%%%%%%%%%%%%%%%%%%%%%%%%%%%%%%%%%%%%%%%%%%

\subsection{Dependence-robust interval estimates}

We present two frequentist methods to construct confidence intervals for the true population size $M^*$ that accommodate weak assumptions about the nature of dependence between samples. The first relies on bootstrap techniques \citep{efron1994introduction} (test inversion bootstrap confidence interval), and the second relies on the properties of profile likelihood ratio statistics (profile likelihood confidence interval) \citep{wilks1938large}. Note that when a statistical model is possibly partially identified instead of point identified, point estimation of a target parameter is not well defined, and thus omitted in our work. We describe the test inversion bootstrap method below, and provide details of profile likelihood confidence intervals in the \emph{Supplementary Appendix}.

\subsubsection{Definition and algorithm}

Our goal is to construct a confidence interval by defining hypothesis tests that depend on $M$, such that the values of $M$ for which the corresponding null hypothesis is rejected will be excluded from the confidence interval. We first use moment inequalities to define the identification set of the true population size $M^*$, and then consider the problem of testing a finite number of moment inequalities. We then invert the test to obtain the confidence interval.  We establish results using the two-step procedure proposed by \citet{romano2014practical}, which has the advantages of controlling the size of the tests uniformly, and remaining computationally feasible when the number of moments is large.  Bootstrap resampling is used to compute critical values for statistical tests. The critical values are a function of the unknown true distribution $P$ of $\mathbf{N}_{\ell}$, and are therefore usually unknown. The basic idea behind the bootstrap approach is that it uses a reasonable approximation to the distribution $P$ to compute critical values. 

Suppose we observe $n$ identically and independently distributed contingency tables, 
$\mathbf{N}_1, \allowbreak \ldots, \allowbreak \mathbf{N}_n \allowbreak  \sim  \allowbreak P \in  \allowbreak \mathcal{P}$. To develop the methodology, consider functionals 
\[  W_{\ell}  = \mathbf{g}\paren{\mathbf{N}_\ell, M} = (g_{1}\paren{\mathbf{N}_\ell, M}, \ldots, g_{\rho}(\mathbf{N}_\ell, M)) \in \mathbb{R}^\rho \]
of the observed data $\mathbf{N}_\ell$ and a given $M$, such that
$ M_I(P)$ is equal to $\lbrace M \in \R_+:\ \E_P[\mathbf{g}\paren{\mathbf{N}_\ell, M} ] \preceq 0 \rbrace, $
where $\E_P[\mathbf{g}\paren{\mathbf{N}_\ell, M} ] \preceq 0$ is called ``moment inequalities''.

We will consider tests of the null hypotheses 
$ H_M: \E_P[\mathbf{g}\paren{\mathbf{N}_\ell, M} ] \preceq 0 ,$ 
that control the probability of a Type I error at level $\alpha$. To illustrate the construction of a confidence set for the true population size $M^*$, we describe the test inversion bootstrap procedure generically in detail in the \emph{Supplementary Appendix}. In short, we formally define the test as
\begin{equation}
\label{eq:phi_n}
\phi_{n}(M, \alpha, \beta)=\left( 1-\mathbbm{1}\left\{Q_{n}(1-\beta, M) \subseteq \mathbb{R}_{-}^{\rho}\right\} \right) \left(1 - \mathbbm{1}\left\{T_{n} \leq \hat{\tau}_{n}(1-\alpha+\beta, M)\right\} \right),
\end{equation}
where $M$ is chosen at the beginning of the algorithm. ($Q_{n}(1-\beta, M), T_n, \hat{\tau}_{n}(1-\alpha+\beta, M)$ are defined formally in the \emph{Supplementary Appendix}). In the implementation, we enumerate $M$ on an arbitrarily fine grid of the positive real line. 
Equation (\ref{eq:phi_n}) states that if either the $1-\beta$ confidence region of $\E_P[\mathbf{g}\paren{\mathbf{N}_\ell, M} ]$, i.e. $Q_{n}(1-\beta, M)$, is a subset of $\mathbb{R}_{-}^{\rho}$, or the test statistics $T_n$ is less than or equal to the critical value $\hat{\tau}_{n}(1-\alpha+\beta, M)$, then we will fail to reject the null hypothesis $H_M$, and therefore this $M$ will remain in the confidence interval. Then, we define the test inversion bootstrap confidence interval as follows:
\begin{defn}[Test Inversion Bootstrap Confidence Interval]
  \label{defn:tibci}
  Fix $0 < \alpha \le 1$, choose any $0< \beta < \alpha$, and let $\phi_{n}(M, \alpha, \beta)$ be defined in \eqref{eq:phi_n}. The test inversion bootstrap (TIB) confidence interval is defined as
  \begin{equation}
  \label{eq:tibci}
    CI_{\text{TIB}}^{n,\alpha}=\{ M \in \R_+ : \phi_{n}(M, \alpha, \beta)=0 \}.
  \end{equation} 
\end{defn}

\subsubsection{TIB confidence intervals under restrictions on pairwise dependence}

Next, we only need to find proper $\mathbf{g}\paren{\mathbf{N}_\ell,  M}$ such that $M_I(P)$ in Equation \eqref{eq:MIP_pair} is equal to $\lbrace M \in \R_+:\ \E_P[\mathbf{g}\paren{\mathbf{N}_\ell, M} ] \preceq 0 \rbrace$.
Let $N_{\ell d_1 d_2}^{(j)} = \sum_{i= 1}^c N_{\ell i} \mathbbm{1} \lbrace { s }_{ir_j} = d_1 \rbrace \mathbbm{1} \lbrace { s }_{it_j} = d_2  \rbrace$, where $d_1, d_2 \in \{0, 1\}$.  Define
\[g_{j1}\paren{\mathbf{N}_\ell,  M} =  - (N_{\ell 11}^{(j)})^2  - N_{\ell 10}^{(j)} N_{\ell 11}^{(j)}  -N_{\ell 01}^{(j)} N_{\ell 11}^{(j)}   -  \xi_{j}N_{\ell 10}^{(j)} N_{\ell 01}^{(j)}  + N_{\ell 11}^{(j)}  +  M N_{\ell 11}^{(j)},\] 
\[g_{j2}\paren{\mathbf{N}_\ell,  M} =     (N_{\ell 11}^{(j)}))^2  + N_{\ell 10}^{(j)} N_{\ell 11}^{(j)}  +N_{\ell 01}^{(j)} N_{\ell 11}^{(j)}   + \eta_{j}N_{\ell 10}^{(j)} N_{\ell 01}^{(j)}  - N_{\ell 11}^{(j)}  -  M N_{\ell 11}^{(j)},\] and \begin{equation}
  W_{\ell}  = \mathbf{g}\paren{\mathbf{N}_\ell,  M}  = \paren{g_{11}\paren{\mathbf{N}_\ell,  M}, g_{12}\paren{\mathbf{N}_\ell,  M}, \ldots, g_{\omega 1}\paren{\mathbf{N}_\ell,  M}, g_{\omega 2}\paren{\mathbf{N}_\ell,  M}} \in \mathbb{R}^{2\omega},
  \label{eq:pair_g}
\end{equation}
which will serve our purposes for moment inequalities. Note that $\mathbf{g}\paren{\mathbf{N}_\ell,  M}$ above is able to deal with possible zero values of $N_{\ell d_1 d_2}^{(j)}$.

\begin{lem}[Moment inequality characterization of the identification set]
  \label{lem:pair_mi}
  Given pairwise restrictions $\paren{r_j, t_j, \eta_{j}, \xi_{j}}_{j = 1}^{\omega}$, $\mathbf{g}\paren{\mathbf{N}_\ell, M} $ defined in \eqref{eq:pair_g} and the parameter space taking the form as in \eqref{eq:Theta_pair}, we have
  \[M_I(P) = \lbrace M \in \R_+:  \E_P\paren{  \mathbf{g}\paren{\mathbf{N}_\ell, M}} \preceq \mathbf{0} \rbrace .\]
\end{lem}

The resulting TIB confidence interval with $\mathbf{g}\paren{\mathbf{N}_\ell, M}$ defined in \eqref{eq:pair_g} will be a uniform asymptotic $1-\alpha$ confidence interval, as summarized below.
\begin{thm}[Asymptotic properties of TIB]
\label{thm:pair_ci}
Under Assumptions \ref{assum:feasibility}, \ref{assum:compactness} and \ref{assum:correctness}, with a parameter space taking the form in \eqref{eq:Theta_pair}, given pairwise restrictions $\paren{r_j, t_j, \eta_{j}, \xi_{j}}_{j = 1}^{\omega}$, and $\mathbf{g}$ defined in \eqref{eq:pair_g}, $ CI_{\text{TIB}}^{n,\alpha}$ defined by \eqref{eq:tibci} satisfies
\begin{equation}
 \liminf _{n \rightarrow \infty} \inf _{P \in \mathcal{P}} \inf _{ M \in M_I(P)} P\left\{ M \in CI_{\text{TIB}}^{n,\alpha}\right\} \geq 1-\alpha,
\end{equation}
where $1-\alpha$ is the pre-specified confidence level.
\end{thm}

\section{Simulations}
\label{subsec:simu}

We investigate the finite-sample performance of the proposed confidence intervals with simulations with $k=3$, $5$, and $10$ capture samples and various restriction forms are available. In addition, we study their performance under violation of assumptions. As a benchmark, we also present CRC population estimates under traditional log-linear Poisson models implemented by popular CRC software packages.
Our simulation results indicate that TIB confidence intervals are always valid, tending to be conservative generally, while PL confidence intervals tend to be more anti-conservative when $k$ or $\omega$ becomes larger. Computationally, the time needed to compute the TIB interval is insensitive to $k$, and grows marginally in proportion to $\omega$, usually within 3 minutes on a MacBook Pro with a 3.1 GHz Dual-Core Intel Core i5 processor. In contrast, the computation time for PL grows quickly with $k$, making it best suited to CRC studies with small $k$.

We first give a short overview of traditional log-linear Poisson models. In the log-linear model, $\log m_i = \blambda'\x_i$, where $\blambda = (\lambda_0, \ldots, \lambda_c) \in\mathbb{R}^{c+1}$ are coefficients for $\x_i$, $i = 0,\ldots,c$. The parameters $\blambda$ have interpretations of dependence between samples. For the population size to be point identified, as in traditional CRC analyses, the analyst must force one or more of the $\lambda_i$ to be zero.  We use the R \citep{R2020} package ``Rcapture'' \citep{baillargeon2007rcapture} to compute confidence intervals under log-linear models. In particular, we use the common \emph{hierarchical} specification of dependence assumptions: if $\lambda_i$ representing interactions among $k'$ samples is set to $0$ ($2 \le k' \le k$), then any higher-order interaction terms involving these $k'$ samples must be $0$.

We show results with 3 samples below. Detailed simulation results for $k = 5, 10$ are available in the \emph{Supplementary Appendix}.
With three simulated capture samples, i.e. $k = 3$, we set the mean value of each observable subset count in our simulations to be equal to the observed empirical counts of PWID for Brussels, Belgium shown in Figure \ref{fig:pwid1}, i.e. $\m^\ast = (21, 103, 13, 89, 29, 24, 27)$, so that the dependence structure resembles real data. We generate one contingency table using Poisson distribution, and compute confidence intervals of the population size using our methods and log-linear models. We repeat the above data generation and estimation for $6000$ times, and compute the frequency of the confidence intervals covering the true population size $M^\ast$, as $M^\ast$ varies.

We compute test inversion bootstrap confidence intervals ($CI_{\text{TIB}}$) and profile likelihood confidence intervals ($CI_{\text{PL}}$) under positive (with $\eta = 1, \xi = 5$) and agnostic (with $\eta = 1/3, \xi = 3$) pairwise dependence restrictions. Here, we define ``\emph{agnostic}'' pairwise restrictions as those which make no assumptions on the direction of dependence, having the form $\eta \le \OR \le \xi$, where $\eta = \frac{1}{\xi}$. 

We also compute confidence intervals under hierarchical log-linear models with all 3 samples and only 2 samples. The hierarchical dependence is described by sample indices (i.e. 1, 2, 3). When these sample indices are not separated by commas, it indicates the existence of an interaction term among these samples, as well as all its nested interaction terms in the log-linear model. 
In practice, epidemiologists often select one model by goodness-of-fit criteria commonly used in regression modeling, such as Akaike information criterion (AIC) \citep{akaike1998information} or Bayesian information criterion (BIC) \citep{schwarz1978estimating}. Here, we use BIC to select the model with the lowest BIC among all models using 3 samples, which is called ``BestBIC" in the following.  The nominal coverage probability is set to $1-\alpha = 0.95$ throughout. 

Figure \ref{fig:simu_part} summarizes the results of our methods and certain log-linear models for comparison (i.e. the independence model [1,2,3], the saturated model [12,13,23], and the BestBIC model. For full results of all hierarchical models, see Figure S1 in \emph{Supplementary Appendix}.)
The horizontal axis is the true population size $M^\ast$ and the vertical axis is the coverage probability for each interval estimate. Two vertical bars mark the identification region of the true population size $M_I(P)$ as in Equation \eqref{eq:MIP_pair}. Therefore, it is clear that whenever $M^\ast$ is in $M_I(P)$, i.e. Assumption \ref{assum:correctness} holds, the coverage probability of our methods by either $CI_{\text{TIB}}$ or $CI_{\text{PL}}$ achieves the nominal value. 
The intervals are generally conservative: actual coverage may be in excess of nominal (i.e. 95\%) coverage for each single value in the identification region. This is because in the partially identified case, many values in the parameter space are observationally equivalently true, and therefore, a valid confidence interval in this case should have the correct coverage probability simultaneously for all these values, i.e. the lowest coverage should at least be 95\%. Due to the similarity of this setting and our real data application, our simulation results imply that the coverage of the confidence intervals for the PWID data set in Section \ref{sec:app} will be close to the nominal level.

In contrast, log-linear models perform less favorably and are not able to flexibly utilize the information of pairwise restrictions. Furthermore, the ``BestBIC'' model cannot achieve the nominal coverage probability even at its peak. Additionally, the average length of confidence intervals by our methods can be comparable to models with strong hierarchical assumptions, with much higher coverage probability at the same time (See Figure S2 in the \emph{Supplementary Appendix}).

We also investigate performances under violation of assumptions. Since Assumption \ref{assum:feasibility} (Feasibility) can be verified, and Assumption \ref{assum:compactness} (Compactness) holds in most cases, we focus on Assumption \ref{assum:correctness} (Correctness). Define $\pmb{m} = \lbrace \m: (\m, M) \in \Theta \text{ for some } M \rbrace$, and  $C_{\m} \equiv \{M: (M, \m) \in \Theta\}$ for a given $\m$.  There are two types of violations: (A) $\m^\ast \in \pmb{m}$, but $M^\ast \notin C_{\m^\ast}$; (B) $\m^\ast \notin \pmb{m}$. In fact, the consequences of type A can be seen in Figure \ref{fig:simu_part}: the farther $M^\ast$ is away from the hypothesized identification region marked by the vertical lines from the misspecified pairwise restrictions, the lower the coverage probability will be for $M^\ast$. Type B violations usually happen when pairwise restrictions are too strong. For example, under our simulation setting, $\eta = 1, \xi = 3$ will render $\m^\ast \notin \pmb{m}$. We study the performance of our methods under this type of violations and the results are shown in Figure S3 in the \emph{Supplementary Appendix}.

\begin{figure}%[h]%[H]
  \centering
  \includegraphics[width=\textwidth]{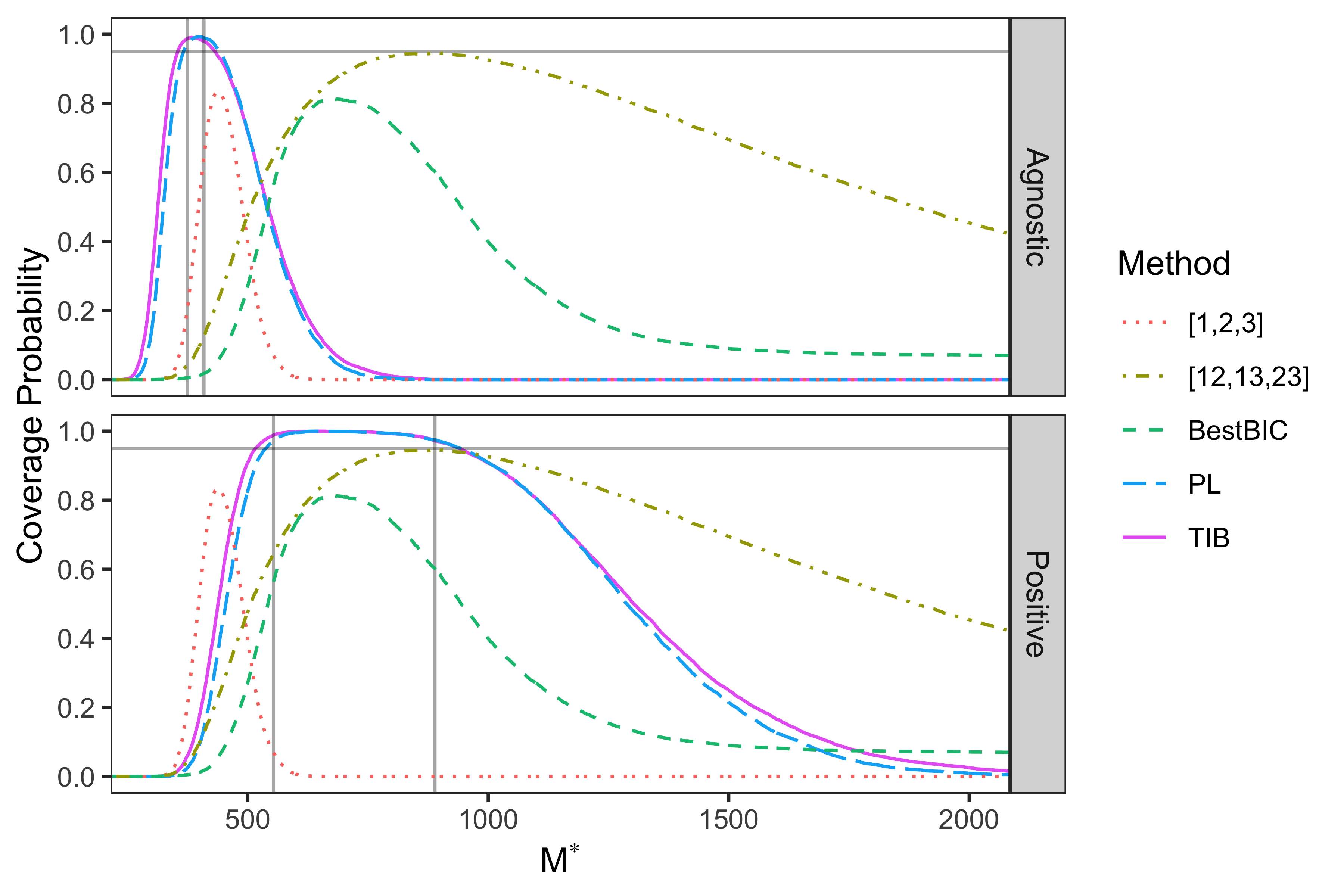}%[width=\textwidth]{simu_part.pdf}
  \caption{Simulation results for coverage probability of nominal 95\% confidence intervals $CI_{\text{TIB}}$ and $CI_{\text{PL}}$ as the true population size $M^\ast$ varies, computed from 6000 simulated tables. Three log-linear models presented as a comparison: the independence model [1,2,3], the saturated model [12,13,23], and BestBIC model. The horizontal line represents the nominal coverage probability $95\%$.  The two vertical lines represent the lower and upper endpoints of identification sets. \emph{Agnostic} pairwise restriction: $\eta = 1/3, \xi = 3$; \emph{Positive} pairwise restriction: $\eta = 1, \xi = 5$.} 
  \label{fig:simu_part}
\end{figure}

\section{Application: estimating the number of people who inject drugs in Brussels, Belgium}
\label{sec:app}

We apply this methodology to a CRC study to estimate the population size of people who inject drugs (PWID) in Brussels, Belgium \citep{plettinckx2020estimates}. Because injection drug use is often stigmatized or legally criminalized, it can be difficult to conduct a systematic survey of PWID \citep{kwon2019estimating}. Instead, indirect estimation techniques like capture-recapture surveys (CRC) are recommended \citep{hay2016estimating}. To update official estimates of the number of PWID (defined here as individuals who injected drugs within the last 12 months) in Brussels to guide the scale and scope of treatment and harm reduction services offered to PWID,  \citet{plettinckx2020estimates} obtained three anonymized PWID samples between February and April 2019 in Brussels from the following sources:
1) an RDS fieldwork study designed to include PWID not in contact with public services \citep{van2020prevalence},
2) two low-threshold drug treatment centers (``MSOC/MASS and Projet Lama'') which offer specialized drug treatment services and opioid substitution treatment, and
3) a crisis intervention center and shelter (``Transit asbl''), which offers psycho-social support during the day and a shelter at night.

The overall subject inclusion criteria, across three data sources, were: having injected any substance within the last 12 months, age 18 or older, and having lived or used drugs in Brussels principally during the last year.  RDS respondents had to be selected by one of the participating organizations as a seed, or have received an invitation by means of a recruitment coupon from a participant, and had not participated earlier.  Violations of the ``closed population'' and ``homogeneous capture probability'' assumptions could happen for the above experiments, however, we assume that they are negligible. We leave robust methods for these violations for future research.

In the Belgium PWID data set, some of the ``seed'' participants in the Respondent Driven Sampling study were from two low-threshold drug treatment centers \citep{van2020prevalence}. Therefore, it is likely that samples 1 and 2 are positively dependent. In addition, since people who approach one service will be more likely to approach another similar service, samples 2 and 3 are also possibly positively dependent. We therefore apply qualitative restrictions on pairwise dependence: all three samples are pairwise positively dependent.

\subsection{Inference under hierarchical log-linear models}

As a comparison, we first show the results of CRC under hierarchical log-linear Poisson models using the R package ``Rcapture'' \citep{baillargeon2007rcapture}. Table \ref{tab:pwid_loglinear} and Figure \ref{fig:rc_high_pair_linear}(a) show point estimates (using the log-linear Poisson model) of the size of the PWID population in Brussels, along with standard errors, 95\% confidence intervals, AIC, and BIC. \citet{plettinckx2020estimates} provide similar estimates using different CRC software. 

\begin{table}[t]%[!htb]
  \centering
  \begin{tabular}{lcllll}
    \toprule 
    Model & $\widehat{M}$ & SE & $CI_{IND}$ & AIC & BIC \\ 
    \midrule
    $[12,13,23]$ & 880 & 293.2 & (505, 1835) & 51.5 & 77.5 \\ 
    $[12,13]$ & 472 & 62.3 & (381, 643) & 62.0 & 84.4 \\ 
    $[12,23]$ & 370 & 21.3 & (336, 421) & 81.4 & 103.8 \\ 
    $[13,23]$ & $\mathbf{688}$ & 97.6 & (535, 936) & 50.3 & $\mathbf{72.7}$ \\ 
    $[12,3]$ & 372 & 18.8 & (340, 414) & 79.5 & 98.1 \\ 
    $[13,2]$ & 530 & 43.0 & (456, 628) & 60.9 & 79.5 \\ 
    $[23,1]$ & 458 & 29.6 & (407, 524) & 91.6 & 110.2 \\ 
    $[1,2,3]$ & 439 & 23.4 & (397, 490) & 92.0 & 106.9 \\ \midrule
    $[1,2]$ & 553 & 54.0 & (463, 679) & 25.0 & 35.9 \\ 
    $[1,3]$ & 272 & 18.2 & (241, 313) & 23.8 & 33.7 \\ 
    $[2,3]$ & 376 & 38.6 & (312, 467) & 24.0 & 34.1 \\ 
    \bottomrule
  \end{tabular}
  \caption{Population size estimates $\widehat{M}$ of the number of people who inject drugs in Brussels, Belgium, standard errors (SE), 95\% confidence intervals $CI_{IND}$, Akaike information criterion (AIC), and Bayesian information criterion (BIC) for point-identified log-linear hierarchical models, computed using the R package ``Rcapture'' Version 1.4-3 \citep{baillargeon2007rcapture}. (Sample 1: Fieldwork Study; Sample 2: Low Threshold Treatment Centers; Sample 3: Crisis Intervention Center and Shelter.) }
  \label{tab:pwid_loglinear}
\end{table}

\begin{table}
  \centering
  \begin{tabular}{llllllll}
    \toprule 
    $\eta_1$ & $\xi_1$ & $\eta_2$ & $\xi_2$ & $\eta_3$ & $\xi_3$ & $CI_{TIB}$ & $CI_{PL}$  \\ 
    \midrule
    % $\mathbf{1}$ & $\mathbf{10}$ & $\mathbf{1}$ & $\mathbf{10}$ & $\mathbf{1}$ & $\mathbf{10}$ & ($\mathbf{436, 1310}$) & ($\mathbf{454, 1284}$) \\
    1 & 10 & 1 & 10 & 1 & 10 & (436, 1310) & (454, 1284) \\
    1 & 5 & 1 & 5 & 1 & 5 & (434, 784) & (454, 754) \\ 
    1 & 3 & 1 & 3 & 1 & 3 & (429, 561) & (425, 617) \\
    0.8 & 10 & 0.8 & 10 & 0.8 & 10 & (404, 1287) & (416, 1243) \\
    1 & 10 & $-\infty$ & $+\infty$ & 1 & 10 & (452, 2699) & (454, 2485)\\
    1 & 10 & $-\infty$ & $+\infty$ & $-\infty$ & $+\infty$ & (469, 3986) & (454, 3887)\\
    1 & 5 & 1 & $+\infty$ & 0.8 & 10 & (451, 2301) & (454, 2195) \\
    \bottomrule 
  \end{tabular}

  \caption{Estimated 95\% confidence intervals for the population size under pairwise dependence restrictions $\eta_1 \le \OR_{1,2} \le \xi_1, \eta_2 \le \OR_{1,3} \le \xi_2, \eta_3 \le \OR_{2, 3} \le \xi_3$. $CI_{TIB}$: Test inversion bootstrap confidence interval. $CI_{PL}$: Profile likelihood confidence interval. }
  \label{tab:pwid_pi}
\end{table}

The ``Model'' column shows the dependence model assumed, where 1, 2, and 3 are sample indices, representing ``Fieldwork Study'', ``Low Threshold Treatment Centers'' and ``Crisis Intervention Center and Shelter'' respectively. When using all three samples, model ``[13,23]'' is the BestBIC model, with an estimate of 688 PWID in Brussels.

Based on recommendations by \citet{hook2000accuracy}, we also investigated interval validity by computing population size estimates only using every two of three samples. These models are labeled as ``[1,2]", ``[1,3]" and ``[2,3]", with results shown in the lower part of Table \ref{tab:pwid_loglinear}. Estimates using only two samples are generally smaller than those using three samples, indicating positive dependence among capture samples. For example, the very low estimate obtained by ``[1,3]" suggests strong positive dependence between samples 1 and 3, the Fieldwork Study sample and the Crisis Intervention Center and Shelter sample.

\subsection{Inference under pairwise restrictions}

We apply our methods using the qualitative information about the pairwise dependence summarized above. Denote pairwise dependence restrictions as $\eta_1 \le \OR_{1,2} \le \xi_1, \eta_2 \le \OR_{1,3} \le \xi_2, \eta_3 \le \OR_{2, 3} \le \xi_3$. Recall that $\OR_{rt}$ is the odds ratio for the capture probabilities in samples $r$ and $t$.
Since all three samples are pairwise positively dependent, we have $\OR_{1,2} = \OR_{1,3} = \OR_{2,3} \in [1, \xi]$ for each pairwise dependence odds ratio.
We choose $\xi_1 = \xi_2 = \xi_3 = \xi = 10$. An odds ratio of 10 is large and conservative, however, it is more credible.
Interval estimates are shown in Table \ref{tab:pwid_pi} and visualized in Figure \ref{fig:rc_high_pair_linear}(b). $CI_{TIB}$ and $CI_{PL}$ have similar estimates. In this case, we use $CI_{TIB}$. Thus, the estimated 95\% confidence interval for the number of people who inject drugs in Brussels, Belgium is between 436 and 1310. 
Additionally, to assess the sensitivity of results to assumptions about strictly positive pairwise dependence, we study results under different values of $\eta_i,\xi_i$ which correspond to different forms of pairwise restrictions. 
The corresponding estimates are shown in Table \ref{tab:pwid_pi} and Figure \ref{fig:rc_high_pair_linear}. In Table \ref{tab:pwid_pi}, Rows 2 and 3 show the influence of the common upper bound $\xi$; Row 4 relaxes the ``positivity'' condition by allowing slightly negative dependence; Rows 5 and 6 show the impact of the number of restrictions; Row 7 allows each restriction to be different.

\begin{figure}%[h]%[H]
  \centering
  \includegraphics[width=\textwidth]{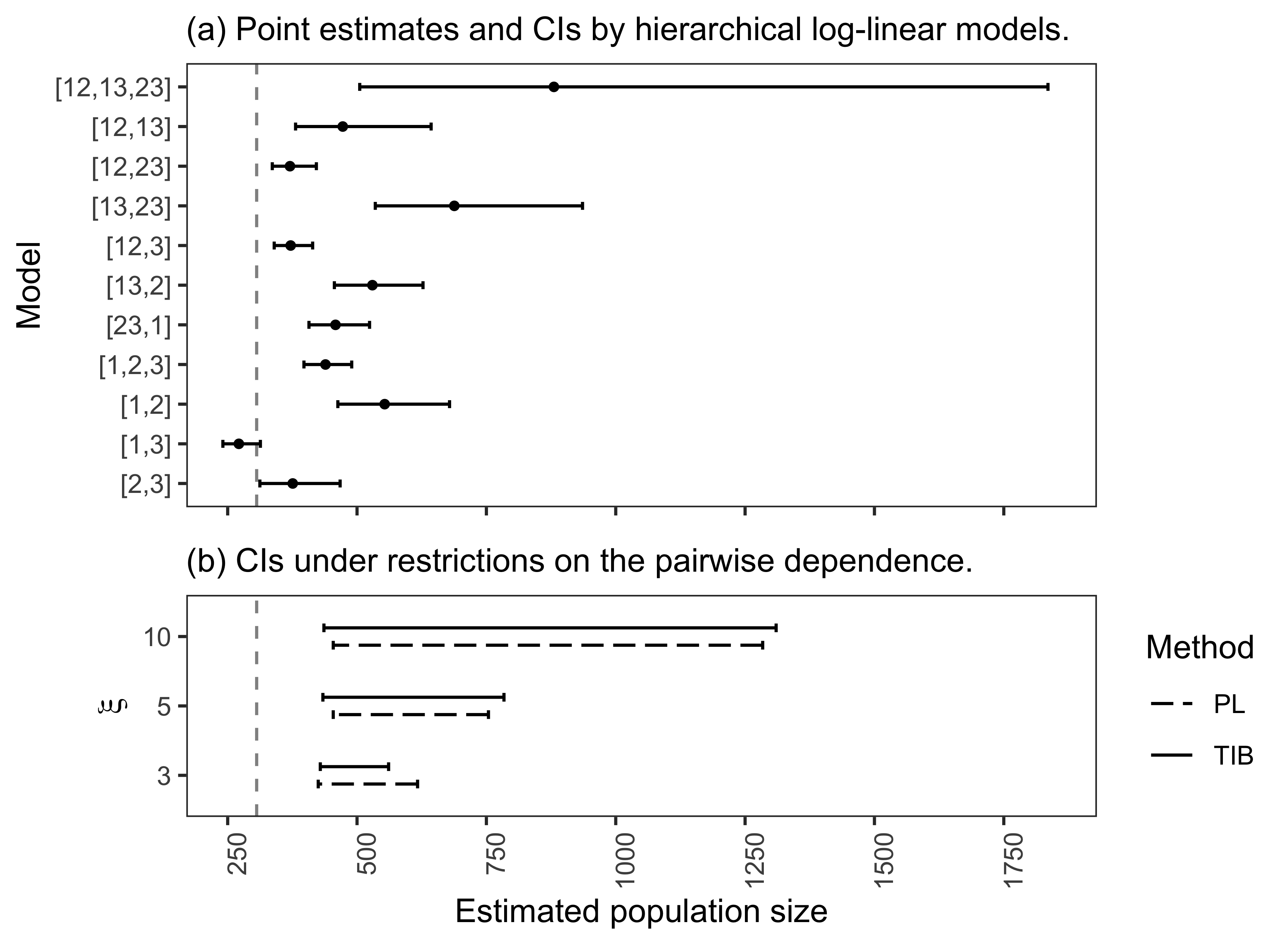}
  \caption{Illustration of 95\% confidence intervals for the number of PWID in Brussels, Belgium.  \textbf{(a)} shows point estimates and confidence intervals (CIs) by hierarchical log-linear models. \textbf{(b)} shows CIs under restrictions on pairwise dependence between samples with $1 \le \OR_{1,2} = \OR_{1,3} = \OR_{2,3} \le \xi$.  The vertical dashed line is a lower bound for the number of PWID in Brussels: the observed number of unique PWID across all three samples, 306. \emph{TIB}: Test inversion bootstrap CIs; \emph{PL}: profile likelihood CIs.}
  \label{fig:rc_high_pair_linear}
\end{figure}

%%%%%%%%%%%%%%%%%%%%%%%%%

\section{Discussion}

CRC surveys are used in situations where experts and policymakers do not agree on the size of the target population. Therefore the empirical credibility of population size estimates hinges on the credibility of the statistical (in)dependence assumptions between samples used.  Usually independent random sampling requires a ``sampling frame'' from which exact or approximate unit sampling probabilities can be computed.  When the size of the target population is truly unknown, the construction of a well-defined sampling frame can be difficult or impossible.  CRC surveys involve several such samples from the target population. When the sampling frame for each survey is ill-defined, it can be difficult to guarantee (in)dependence structures of samples, or to estimate the nature of dependence between samples.  

In this paper, we proposed a novel frequentist method that is flexible in the identification phase of the statistical problem of CRC experiments, which allows easy incorporation of domain knowledge on pairwise dependence. Our inferential procedures are theoretically valid to deal with possibly partially identified parameters, and thus eliminate the need to employ implausible assumptions to achieve point identification.  Our approach is distinct from that of Bayesian approaches in CRC \citep{aleshin2021revisiting} because it does not impose a prior distribution over possibly partially identified parameters; rather, we assume that only bounds on these parameters are known from empirical knowledge. By the general Bayesian theory of partial identification \citep{canay2017practical,  moon2012bayesian, kitagawa2012estimation}, prior information for partially identified parameters will not be washed out, even asymptotically. Therefore, Bayesian credible intervals tend to be shorter than frequentist confidence intervals in this case because they retain information from the prior over dependence parameters, even in large samples.

This work has several limitations. First, we have focused on the case of homogeneous sampling probabilities within samples. Since differing selection probabilities may occur in complex CRC surveys \citep{gimenez2018individual}, in ongoing work extending the approach presented here, we use measured categorical covariates (e.g. sex and age groups) to deal with heterogeneity. We will stratify all the subjects by covariates, and impose stratum-specific pairwise restrictions from empirical knowledge. Applying our partial identification framework to the combination of the sets of moment conditions from each stratum results in an interval estimate of the total population size.
Second, we have not addressed more complex knowledge of putative dependence structures beyond pairwise relationships. Although rare, when this type of information is available, it leads to additional moment conditions that can be easily incorporated into the current methodological framework. 

In general, we recommend that researchers rely on domain knowledge of the target population and the nature of the sampling procedures to choose $\eta$ and $\xi$. When such information is vague, we suggest choosing a conservative value to ensure the credibility of inference results. When domain knowledge is plentiful (e.g. all samples are independent, or samples are pairwise positively dependent), the methodology proposed here will deliver highly informative (narrow) interval estimates for the target population size, as exemplified in the Application section. 
We expect that generally, inferences based on the weakest credible assumptions may be more useful to empiricists or policymakers who may not agree on the exact nature and magnitude of dependence between samples.

\if0\blind
{
\noindent \textbf{Acknowledgements}: This work was supported by NIH grant NICHD DP2 HD091799-01.  We are grateful to 
 P. M. Aronow
 and
 Si Cheng
for helpful comments on the manuscript. We thank the local partners Transit asbl, MASS de Bruxelles, Projet Lama and SamuSocial for their support in reaching out to PWID, as well as Lies Gremeaux and J\'er\^ome Antoine for their involvement in the fieldwork. We thank the nurses who conducted the fieldwork and who had a major contribution to the success of the study. Last but not least thanks to all the participants for their confidence and time. 
} \fi

\textbf{Supplementary Appendix}: Proofs and additional lemmas, details of constructing test inversion bootstrap confidence intervals, and extra simulation results are available in the \textit{Supplementary Appendix} online. We implemented the proposed methodology in the R package \pkg{crc.partialid} available at \url{https://github.com/Jinghao-Sun/crc.partialid}, which also includes the Brussels PWID data set.

\bibliographystyle{agsm}
\bibliography{crc_outline}

\end{document}

\maketitle
\tableofcontents
\newpage
\section{Proofs}
% \commentjs{TODO: remove referee comments, supp files(figs, captions, proof of theorem), edit abstract}
% \commentjs{TODO: Cite bishop in lemma 1 proof?  In supp, for each proof, use plain language to describe the proof, and why it works. Cite recent COVID partial ID paper.}
% \comment{wide interval, choose parameters, novelty compared to bayesian approach. Suggest reviewers like Gerritse, S. C., van der Heijden, P. G. and Bakker, B. F. , Zhang, Z. Even so, Bernie could still use a bayesian glass to look at our paper. Reduce the amount of exposition for highest order restrictions? Cite recent frequentist discussion on identifiability. Many such papers are about heterogeneous capture probability.}

First, we list three assumptions for easy reference, which will be assumed throughout.
\begin{assum}[Feasibility]
\label{assum:feasibility}
The parameter space $\Theta \subseteq \R_{+}^{c+1}$ is non-empty.
\end{assum}

\begin{assum}[Compactness]
  \label{assum:compactness}
  The parameter space $\Theta$ is closed, and bounded away from 0 and $\infty$.
\end{assum}

\begin{assum}[Correctness]
\label{assum:correctness}
The true parameter $\theta^\ast \in \Theta$.
\end{assum}

\subsection{Lemma \ref{lem:Theta_pair}}
Lemma \ref{lem:Theta_pair} describes the parameter space and identification set under pairwise dependence restrictions.
\begin{lem}[Identification set of $M^*$ under restrictions on pairwise dependence]
  \label{lem:Theta_pair}
  Given pairwise restrictions $\paren{r_j, t_j, \eta_{j}, \xi_{j}}_{j = 1}^{\omega}$, the model space $\mathcal{P}$ has corresponding parameter space
  \begin{equation}
  \label{eq:Theta_pair}
  \begin{split}
  \Theta =\Bigg\{ & (M, \m) \in \mathbb{R}_+^{c+1}: \eta_j \le \frac {m_{11}^{(j)} \left[ M - m_{10}^{(j)}  - m_{01}^{(j)}  - m_{11}^{(j)} \right]}{m_{10}^{(j)} m_{01}^{(j)} } \le \xi_j, \text{ for } j=1,\ldots,\omega   \Bigg\} .
  \end{split}
  \end{equation}
  When the true parameter vector is $\theta^\ast = (\m^\ast, M^\ast)$, define $m_{d_1 d_2}^{\ast (j)} = \sum_{i= 1}^c m_i^{\ast} \mathbbm{1} \lbrace { s }_{ir_j} = d_1 \rbrace \mathbbm{1} \lbrace { s }_{it_j} = d_2  \rbrace $. Then the identification set for $M^*$ is 
  \begin{equation}
  \label{eq:MIP_pair}
  \begin{split}
  M_I(P) = \Bigg[ 
  & \max_{j \in \{1, \ldots, \omega \}} \left\lbrace \eta_j \frac{m_{10}^{\ast (j)}m_{01}^{\ast (j)}}{m_{11}^{\ast (j)}} + m_{10}^{\ast (j)} + m_{01}^{\ast (j)} + m_{11}^{\ast (j)}
   \right\rbrace, \\
  & \min_{j \in \{1, \ldots, \omega \}}  \left\lbrace \xi_j \frac{m_{10}^{\ast (j)}m_{01}^{\ast (j)}}{m_{11}^{\ast (j)}} + m_{10}^{\ast (j)} + m_{01}^{\ast (j)} + m_{11}^{\ast (j)} \right\rbrace
  \Bigg] .
  \end{split}
  \end{equation}
\end{lem}

\begin{proof}[Proof of Lemma \ref{lem:Theta_pair}]
Given pairwise restrictions $\paren{r_j, t_j, \eta_{j}, \xi_{j}}_{j = 1}^{\omega}$, under Assumption \ref{assum:compactness}, the expression 
\begin{equation}
  \begin{split}
  \Theta =\Bigg\{ & (M, \m) \in \mathbb{R}_+^{c+1}: \eta_j \le \frac {m_{11}^{(j)} \left[ M - m_{10}^{(j)}  - m_{01}^{(j)}  - m_{11}^{(j)} \right]}{m_{10}^{(j)} m_{01}^{(j)} } \le \xi_j, \text{ for } j=1,\ldots,\omega   \Bigg\} .
  \end{split}
\end{equation}
follows from the facts that the OR between list $r_j, t_j$ equals $\frac {m_{11}^{(j)} \left[ m_0 + m_{00}^{(j)} \right]}{m_{10}^{(j)} m_{01}^{(j)} }$ and $M = m_0 + m_{00}^{(j)}  + m_{10}^{(j)}  + m_{01}^{(j)}  + m_{11}^{(j)}$, for all $j$.  Then, by Assumptions \ref{assum:feasibility} and \ref{assum:correctness}, we obtain the form of $M_I(P)$.
\end{proof}

\subsection{Lemma \ref{lem:pair_mi}}

Lemma \ref{lem:pair_mi} shows that the specified moment inequalities are equivalent to the identification set.
\begin{lem}[Moment inequality characterization of the identification set]
  \label{lem:pair_mi}
  Given pairwise restrictions $\paren{r_j, t_j, \eta_{j}, \xi_{j}}_{j = 1}^{\omega}$, $\mathbf{g}\paren{\mathbf{N}_\ell, M} $ defined in \eqref{eq:pair_g} and the parameter space taking the form as in \eqref{eq:Theta_pair}, we have
  \[M_I(P) = \lbrace M \in \R_+:  \E_P\paren{  \mathbf{g}\paren{\mathbf{N}_\ell, M}} \preceq \mathbf{0} \rbrace .\]
\end{lem}

\begin{proof}[Proof of Lemma \ref{lem:pair_mi}]
Note that $N_{\ell 11}^{(j)}, N_{\ell 10}^{(j)}, N_{\ell 01}^{(j)}$ are independent Poisson random variables with mean $m_{11}^{\ast (j)}$, $m_{10}^{\ast (j)}$, $m_{01}^{\ast (j)}$, where $m_{d_1 d_2}^{\ast (j)} = \sum_{i= 1}^c m_i^{\ast} \mathbbm{1} \lbrace \boldsymbol{ s }_{ir_j} = d_1 \rbrace \mathbbm{1} \lbrace \boldsymbol{ s }_{it_j} = d_2  \rbrace $. Thus, \[ \E_P(g_{j1}\paren{\mathbf{N}_\ell, M}) \le 0 \Rightarrow \frac{m_{11}^{\ast (j)}\left( M - m_{10}^{\ast (j)} - m_{01}^{\ast (j)} - m_{11}^{\ast (j)} \right)}{m_{10}^{\ast (j)}m_{01}^{\ast (j)}} \le \xi_j ,\]
\[ \E_P(g_{j2}\paren{\mathbf{N}_\ell, M}) \le 0 \Rightarrow \frac{m_{11}^{\ast (j)}\left( M - m_{10}^{\ast (j)} - m_{01}^{\ast (j)} - m_{11}^{\ast (j)} \right)}{m_{10}^{\ast (j)}m_{01}^{\ast (j)}} \ge \eta_j .\] 
% Thus, compared to Equation (\ref{eq:MIP_pair}), 
Then, it is easy to verify $M_I(P) = \lbrace M \in \R_+:  \E_P\paren{  \mathbf{g}\paren{\mathbf{N}_\ell, M}} \preceq \mathbf{0} \rbrace$.
\end{proof}

% See \textit{Supplementary Appendix}.
% [Proof of Lemma \ref{lem:high_regularity}]

\subsection{Lemma \ref{lem:pair_regularity}}

Lemma \ref{lem:pair_regularity} shows that $\mathbf{g}$ in \eqref{eq:pair_g} is ``well-behaved". This is a technical regularity condition needed to prove the main theorem below. We first introduce some notations.

Define $N_{\ell d_1 d_2}^{(j)} = \sum_{i= 1}^c N_{\ell i} \mathbbm{1} \lbrace \boldsymbol{ s }_{ir_j} = d_1 \rbrace \mathbbm{1} \lbrace \boldsymbol{ s }_{it_j} = d_2  \rbrace$, where $d_1, d_2 \in \{0, 1\}$, and 
\[ g_{j1}\paren{\mathbf{N}_\ell,  M} =  - (N_{\ell 11}^{(j)})^2  - N_{\ell 10}^{(j)} N_{\ell 11}^{(j)}  -N_{\ell 01}^{(j)} N_{\ell 11}^{(j)}   -  \xi_{j}N_{\ell 10}^{(j)} N_{\ell 01}^{(j)}  + N_{\ell 11}^{(j)}  +  M N_{\ell 11}^{(j)},\]
\[ g_{j2}\paren{\mathbf{N}_\ell,  M} =     (N_{\ell 11}^{(j)}))^2  + N_{\ell 10}^{(j)} N_{\ell 11}^{(j)}  +N_{\ell 01}^{(j)} N_{\ell 11}^{(j)}   + \eta_{j}N_{\ell 10}^{(j)} N_{\ell 01}^{(j)}  - N_{\ell 11}^{(j)}  -  M N_{\ell 11}^{(j)}, \]
Let \begin{equation}
  W_{\ell}  = \mathbf{g}\paren{\mathbf{N}_\ell,  M}  = \paren{g_{11}\paren{\mathbf{N}_\ell,  M}, g_{12}\paren{\mathbf{N}_\ell,  M}, \ldots, g_{\omega 1}\paren{\mathbf{N}_\ell,  M}, g_{\omega 2}\paren{\mathbf{N}_\ell,  M}}	\in \mathbb{R}^{2\omega} .
  \label{eq:pair_g}
\end{equation}

Define $g_j$ as the $j$th component of $\mathbf{g}$, $j = 1, \ldots, 2\omega$. Let $\mu_j(P, M) = \E_P[g_j(\mathbf{N}_\ell,  M)]$, and $\sigma^2_j(P, M) = \Var_P[g_j(\mathbf{N}_\ell,  M)]$. Then, we need the uniform integrability condition below.

\begin{lem}[Uniform integrability of $\mathbf{g}$ for restrictions on pairwise dependence]
  \label{lem:pair_regularity}
  For each $j$,
  \[ \lim_{\kappa \rightarrow \infty} \sup_{P \in \mathcal{P}} \sup_{ M \in M_I(P)} \E_P \left[
  \paren{\frac{g_j(\mathbf{N}_\ell, M) - \mu_j(P, M)}{\sigma_j(P, M)}}^2\mathbbm{1}\paren{\vert \frac{g_j(\mathbf{N}_\ell, M) - \mu_j(P, M)}{\sigma_j(P, M)} \vert > \kappa}
  \right] = 0 .\]
\end{lem}

\begin{proof}[Proof of Lemma \ref{lem:pair_regularity}]
  For ease of notation, in this proof, let $g_j \equiv g_j(\mathbf{N}_\ell, M), \mu_j \equiv \mu_j(P, M), \sigma_j \equiv \sigma_j(P, M)$. For a given $\kappa$, $P$, and $M \in M_I(P)$, for all $j$ we have, 
  \begin{align*}
    &\E_P \left[
    \paren{\frac{g_j - \mu_j}{\sigma_j}}^2\mathbbm{1}\paren{\vert \frac{g_j - \mu_j}{\sigma_j} \vert > \kappa} \right] \\
    &= \E_P \left[
    \paren{\frac{g_j - \mu_j}{\sigma_j}}^2 \right] \E_P\left[\mathbbm{1}\paren{\vert \frac{g_j - \mu_j}{\sigma_j} \vert > \kappa} \right] + \Cov_P \left[\paren{\frac{g_j - \mu_j}{\sigma_j}}^2,   \mathbbm{1}\paren{\vert \frac{g_j - \mu_j}{\sigma_j} \vert > \kappa}\right] \\
    & \text{(By Cauchy-Schwarz inequality)}\\
    & \le P \left\{\vert \frac{g_j - \mu_j}{\sigma_j} \vert > \kappa \right\}  + \sqrt{\Var_P\left[\paren{\frac{g_j - \mu_j}{\sigma_j}}^2 \right] \Var_P\left[ \mathbbm{1}\paren{\vert \frac{g_j - \mu_j}{\sigma_j} \vert > \kappa} \right]} \\
    & = P \left\{\vert \frac{g_j - \mu_j}{\sigma_j} \vert > \kappa \right\}  + \sqrt{P \left\{\vert \frac{g_j - \mu_j}{\sigma_j} \vert > \kappa \right\} \left( 1 - P \left\{\vert \frac{g_j - \mu_j}{\sigma_j} \vert > \kappa \right\} \right) \Var_P\left[ \paren{\frac{g_j - \mu_j}{\sigma_j}}^2 \right]} \\
    & \le  P \left\{\vert \frac{g_j - \mu_j}{\sigma_j} \vert > \kappa \right\}  + \sqrt{P \left\{\vert \frac{g_j - \mu_j}{\sigma_j} \vert > \kappa \right\}  \E_P\left[ \paren{\frac{g_j - \mu_j}{\sigma_j}}^4 \right]}\\
    & \text{(By Chebyshev inequality)}\\
    & \le \frac{1}{\kappa^2}  + \frac{1}{\kappa}\sqrt{  \E_P\left[ \paren{\frac{g_j - \mu_j}{\sigma_j}}^4 \right]}
  \end{align*}
  Note that for all $j$, since $N_{\ell 11}^{(j)}, N_{\ell 01}^{(j)}, N_{\ell 10}^{(j)}$ are independent Poisson random variables which have any finite order moments that are continuous functions of $\m$,  then $g_j$ has any finite order moments that are continuous functions of $(M, \m)$. Since \[\E_P\left[ \paren{\frac{g_j - \mu_j}{\sigma_j}}^4 \right] = \frac{\sum_{k = 0}^4 {4 \choose k} (-\mu_j)^k \E(g_j^{4-k})}{{\sigma_j}^4}, \] then $\E_P\left[ \paren{\frac{g_j - \mu_j}{\sigma_j}}^4 \right]$ is a continuous function in the underlying parameter $(M, \m)$.

  Therefore, we have 
  \begin{align*}
    & \lim_{\kappa \rightarrow \infty} \sup_{P \in \mathcal{P}} \sup_{ M \in M_I(P)} \E_P \left[
      \paren{\frac{g_j - \mu_j}{\sigma_j}}^2\mathbbm{1}\paren{\vert \frac{g_j - \mu_j}{\sigma_j} \vert > \kappa} \right] \\
    &\le \lim_{\kappa \rightarrow \infty} \sup_{\theta \in \Theta} \E_{P_\theta} \left[
      \paren{\frac{g_j - \mu_j}{\sigma_j}}^2\mathbbm{1}\paren{\vert \frac{g_j - \mu_j}{\sigma_j} \vert > \kappa} \right]\\
    &\le \lim_{\kappa \rightarrow \infty} \frac{1}{\kappa^2}  + \frac{1}{\kappa}\sqrt{  \sup_{\theta \in \Theta}  \E_P\left[ \paren{\frac{g_j - \mu_j}{\sigma_j}}^4 \right]}\\
    & \text{(By extreme value theorem, and that $\Theta$ is compact and the continuity established above,}\\
    & \text{$\E_P\left[ \paren{\frac{g_j - \mu_j}{\sigma_j}}^4 \right]$ is bounded away from $\infty$.)}\\
    & = 0.
  \end{align*}

\end{proof}

\subsection{Theorem \ref{thm:pair_ci}}

\begin{thm}[Asymptotic properties of TIB]
  \label{thm:pair_ci}
  Under Assumptions \ref{assum:feasibility}, \ref{assum:compactness} and \ref{assum:correctness}, with a parameter space taking the form in \eqref{eq:Theta_pair}, given pairwise restrictions $\paren{r_j, t_j, \eta_{j}, \xi_{j}}_{j = 1}^{\omega}$, and $\mathbf{g}$ defined in \eqref{eq:pair_g}, $ CI_{\text{TIB}}^{n,\alpha}$ defined by 
  \[CI_{\text{TIB}}^{n,\alpha}=\{ M \in \R_+ : \phi_{n}(M, \alpha, \beta)=0 \}\]
  (Equation (4) in the main text) satisfies
  \begin{equation}
   \liminf _{n \rightarrow \infty} \inf _{P \in \mathcal{P}} \inf _{ M \in M_I(P)} P\left\{ M \in CI_{\text{TIB}}^{n,\alpha}\right\} \geq 1-\alpha,
  \end{equation}
  where $1-\alpha$ is the pre-specified confidence level.
\end{thm}
\begin{proof}
  This follows by Lemma \ref{lem:pair_mi} and \ref{lem:pair_regularity}, and Theorem 3.1 of \citet{romano2014practical}.
\end{proof}

\section{Detailed description of test inversion bootstrap confidence interval}

To illustrate the construction of a confidence set for the true population size $M^*$, we describe the test inversion bootstrap procedure generically.  First, choose $M \in \R_+$ and note that $W_{\ell}  = \mathbf{g}\paren{\mathbf{N}_\ell, M}$ is a function of the given value of $M$, for $\ell = 1, \ldots, n$.  We will consider tests of the null hypotheses 
$ H_M: \E_P[\mathbf{g}\paren{\mathbf{N}_\ell, M} ] \preceq 0 ,$ 
that control the probability of a Type I error at level $\alpha$.

Let \(\mu(P, M)\) denote the mean of $W_{\ell}$ under \(P\), i.e. $\mu(P, M) = \E_P[\mathbf{g}\paren{\mathbf{N}_\ell, M} ]$, and let \(\mu_{j}(P, M)\) denote the \(j\)th component of \(\mu(P, M)\). Let $\tilde{P}_n$ denote the joint distribution of $c$ independent Poisson distribution with means $\paren{\bar{N}_{(n)1}, \ldots, \bar{N}_{(n)c}}$. Here, $\tilde{P}_n$ is an approximation to the true distribution $P$, and will be used to generate samples to compute critical values. This particular form of the bootstrap is called parametric bootstrap, because it invokes the Poisson distributional assumption for $\mathbf{N_{\ell}}$. Note that a Poisson random variable only has one unknown parameter, and this is why in practice we can specify $\tilde{P}_n$ when $n$ is almost always equal to 1. Define \(\bar{W}_{n}=\mu\left(\tilde{P}_{n}, M\right)\) and \(\bar{W}_{j, n}=\mu_{j}\left(\tilde{P}_{n}, M\right)\), the sample analogues of \(\mu(P, M)\) and \(\mu_{j}(P, M)\). The notation  \(\sigma_{j}^{2}(P, M)\) denotes the variance of the $j$th component of $W_{\ell}$ under \(P\) and the sample analogue is \(S_{j, n}^{2}=\sigma_{j}^{2}\left(\tilde{P}_{n}, M\right)\).  %Finally, let \(S_{n}^{2}=\operatorname{diag}\left(S_{1, n}^{2}, \ldots, S_{\rho, n}^{2}\right)\).  
Define the test statistics as
\begin{equation}
T_{n}= \max_{1 \leq j \leq \rho}  \frac{\sqrt{n} \bar{W}_{j,n}}{S_{j,n}} .
\label{eq:Tn}
\end{equation}
Intuitively, a large positive value of $T_n$ suggests a deviation from the null hypothesis, and when it exceeds a critical value, we will reject the null. Next, we calculate a critical value for $T_{n}$ using bootstrap techniques. 

First, generate $B$ iid random samples from the empirical distribution $\tilde{P}_{n}$, each with size $n$. Denote the $b$th sample as $\{ \mathbf{N}_\ell^{b}\}_{\ell = 1,\ldots, n}$, for $b = 1,\ldots,B$. There are two steps in the bootstrap procedure. Step 1 is used to determine which components of $\mu(P, M)$ are ``negative'', and this information is incorporated in Step 2 to reduce the computational burden to compute the critical value.

In \emph{Step 1}, compute $Q_{n}(1-\beta, M)$, a bootstrap $(1 - \beta)$ confidence region for $\mu(P, M)$ with $\beta \in (0, \alpha)$, as
\begin{equation}
\label{eq:Qn}
Q_{n}(1-\beta, M)=\left\{u = (u_1, \ldots, u_\rho) \in \mathbb{R}^{\rho}: \max _{1 \leq j \leq \rho} \frac{\sqrt{n}\left(u_{j}-\bar{W}_{j, n}\right)}{S_{j, n}} \leq K_{n}^{-1}\left(1-\beta, \tilde{P}_{n}, M\right)\right\} ,
\end{equation}
where
\begin{equation}
\label{eq:Kn}
K_{n}(x, P, M)=P\left\{\max _{1 \leq j \leq \rho} \frac{\sqrt{n}\left(\mu_{j}(P, M)-\bar{W}_{j, n}\right)}{S_{j, n}} \leq x\right\} .
\end{equation}
Here, $K_n(x, P, M)$ is the Cumulative Distribution Function (CDF) of $\max _{1 \leq j \leq \rho} \frac{\sqrt{n}\left(\mu_{j}(P, M)-\bar{W}_{j, n}\right)}{S_{j, n}} $ under the true distribution of $\mathbf{N_{\ell}}$, and $K_n^{-1}\left(1-\beta, \tilde{P}_{n}, M\right)$ is the approximated $1-\beta$ quantile of $K_n(x, P, M)$ computed using parametric bootstrap.

The critical value $K_{n}^{-1}\left(1-\beta, \tilde{P}_{n}, M\right)$ is obtained as follows.  For the $b$th bootstrap sample $\{ \mathbf{N}_\ell^{b}\}, \ell = 1,\ldots, n$, define $\bar{N}_{(n)j}^{b} = \frac{1}{n} \sum_{\ell = 1}^n N_{\ell j}^{b}$. Let $\tilde{P}_n^{b}$ denote the joint distribution of $c$ independent Poisson distribution with mean $\paren{\bar{N}_{(n)1}^{b}, \ldots, \bar{N}_{(n)c}^{b} }$, and \(\bar{W}_{j, n}^{b }=\mu_{j}\left(\tilde{P}_{n}^{b }, M\right) ,\) \(S_{j, n}^{b, 2}=\sigma_{j}^{2}\left(\tilde{P}_{n}^{b }, M\right)\). Then, $K_{n}^{-1}\left(1-\beta, \tilde{P}_{n}, M\right)$ can be approximated arbitrarily well with large $B$, as the $1-\beta$ quantile of 
$ \left\{ \max _{1 \leq j \leq \rho} \frac{\sqrt{n}\left(\bar{N}_{(n)j}-\bar{W}_{j, n}^{b}\right)}{S_{j, n}^{b}} \right\}_{b = 1}^B. $
As suggested by simulations in \citet{romano2014practical}, $\beta$ can be conveniently set as $\alpha/10$, since larger values of $\beta$ reduce statistical power in general, and lower values of $\beta$ require more bootstrap samples, and thus increase the computational burden.

In \emph{Step 2}, for \(x \in \mathbb{R}\) and \(\bzeta = (\zeta_1, \ldots, \zeta_\rho) \in \mathbb{R}^{\rho},\) let
% \label{eq:An}
\[ A_{n}(x, \bzeta, P, M)=P\left\{\max_{1 \leq j \leq \rho}  \sqrt{n}  \frac{\bar{W}_{j,n} - \mu_j(P, M) + \zeta_j}{S_{j,n}}  \leq x\right\} . \]
Note that 
\begin{equation}
\label{eq:PT_n}
P\{T_n \le x\} = A_n(x, \mu(P, M), P, M),
\end{equation}
and for any $x, P, M$, $A_{n}(x, \bzeta, P, M)$ is non-increasing in each component of $\bzeta$. Thus, the $1-\alpha$ quantile of $A_n(x, \mu(P, M), P, M)$, $A_n^{-1}(1-\alpha, \mu(P, M), P, M)$, is the critical value of interest for our test statistics $T_n$. To approximate this value, it is tempting to replace $P$ in the right-hand side of \eqref{eq:PT_n} with $\tilde{P}_{n}$. However, this approximation to the distribution of $T_n$ fails when the true parameter $\theta^*$ is at the boundary of $\Theta$, as shown by \citet{andrews2000inconsistency}.  A provably valid alternative approach, as shown in \citep{romano2014practical}, is to find a $\bzeta^* = (\zeta^*_1, \ldots, \zeta^*_\rho)$, such that $\bzeta^* \succeq \mu(P, M)$ with high probability, and replace $\mu(P, M)$ with $\bzeta^*$ in \eqref{eq:PT_n}, then replace the second $P$ in \eqref{eq:PT_n} with $\tilde{P}_{n}$, as shown in \eqref{eq:tau_n}. 

In \eqref{eq:Qn}, we have computed the $1-\beta$ confidence region for $\mu(P, M), P, M)$. Also, under the null hypothesis, $\mu(P, M) \preceq 0$. Based on these two facts, define
%\label{eq:zeta}
\[ \zeta_{j}^{*}=\min \left\{\bar{W}_{j, n}+\frac{S_{j, n} K_{n}^{-1}\left(1-\beta, \tilde{P}_{n}, M\right)}{\sqrt{n}}, 0\right\} , \]
which is the upper limit of the confidence interval for $\mu_j(P, M)$ suggested by \eqref{eq:Qn} after being truncated at 0, and 
\begin{equation}
\label{eq:tau_n}
\hat{\tau}_{n}(1-\alpha+\beta, M)=A_{n}^{-1}\left(1-\alpha+\beta, \bzeta^{*}, \tilde{P}_{n}, M\right) ,
\end{equation}
which is the $1-\alpha + \beta$ quantile of the bootstrap sample
\[ \left\{  \max _{1 \leq j \leq \rho} \sqrt{n}  \paren{\bar{W}_{j, n}^{b} - \bar{N}_{(n)j} +\zeta_j^{*} }/ S_{j, n}^{b}  \right\}_{b = 1}^B. \] 
Since $A_{n}(x, \bzeta, P, M)$ is non-increasing in each component of $\bzeta$, 
\[ A_n^{-1}(1-\alpha, \mu(P, M), P, M) \le A_n^{-1}(1-\alpha+\beta, \bzeta^*, \tilde{P}_{n}, M) \] 
with high probability, where the added $\beta$ is to account for the possibility that $\mu(P, M)$ lies outside of $Q_n(1-\beta, M)$. Thus, we will use $\hat{\tau}_{n}(1-\alpha+\beta, M)$ as our critical value for $T_n$. Finally, we formally define our test as
\begin{equation}
\label{eq:phi_n_S}
\phi_{n}(M, \alpha, \beta)=\left( 1-\mathbbm{1}\left\{Q_{n}(1-\beta, M) \subseteq \mathbb{R}_{-}^{\rho}\right\} \right) \left(1 - \mathbbm{1}\left\{T_{n} \leq \hat{\tau}_{n}(1-\alpha+\beta, M)\right\} \right),
\end{equation}
where $M$ is chosen at the beginning of the algorithm. Equation (\ref{eq:phi_n_S}) states that if either the $1-\beta$ confidence region $Q_{n}(1-\beta, M)$ is a subset of $\mathbb{R}_{-}^{\rho}$, or $T_n$ is less than or equal to the critical value, then we will fail to reject the null hypothesis $H_M$, and therefore this $M$ will remain in the confidence interval. Then we can formally define our test inversion bootstrap confidence interval as in Definition 1 of the main text.

\section{Detailed description of the profile likelihood confidence interval}

In this section we introduce an alternative algorithm to construct interval estimates for $M^*$, which is based on the profile likelihood function of $M$. For a sequence of $n$ iid CRC studies under the same sampling design, we define the \emph{sample criterion function} as 
$ L_n(\theta) = \sum_{i = 1}^c -m_i + \bar{N}_{(n)i}\log m_i ,$ where $\theta=(M,\mathbf{m})$, 
and define 
$ \widehat{\Theta}_I = \left\{ \hat{\theta}: L_n(\hat{\theta}) = \sup_{\theta \in \Theta} L_n(\theta) \right\} .$

In point-identified parametric models, confidence intervals constructed by inverting a likelihood ratio test have been well studied \citep{wilks1938large, van2000asymptotic}, and may have better finite-sample coverage properties than their Wald-type counterparts. For CRC surveys, we only wish to conduct inference for the population size $M$ and consider $\m$ as nuisance parameters. Thus, by profiling out $\m$, the profile likelihood ratio statistics defined below are only a function of $M$.
% chernoff1954distribution

\begin{defn}[Profile likelihood confidence interval]
\label{defn:plci}

Define the likelihood-ratio as $\mathcal{LR}_n(\theta) = 2n[L_n(\hat\theta) - L_n(\theta)]$, where $\hat\theta \in \widehat{\Theta}_I$ and let $C_M = \lbrace \m: (M, \m) \in \Theta \rbrace$ for a given $M$. Then the profile likelihood-ratio statistics is $\mathcal{PLR}_n(M) \equiv \inf_{\m \in C_M} \mathcal{LR}_n(M, \m)$. Define set 
\begin{equation}
CI_{\text{PL}} = \lbrace M \in \R_+: \mathcal{PLR}_n(M) \leq \chi^2_{1, \alpha} \rbrace ,
\end{equation}
where $\chi^2_{1, \alpha}$ is the $1-\alpha$ quantile of the chi-square distribution with one degree of freedom (i.e. $\chi^2_1$), as a Profile Likelihood Confidence Interval, which is hypothesized to be an asymptotic $1-\alpha$ confidence interval. 
\end{defn}

The computation procedure is straightforward. First, find any $\theta$ that maximizes $L_n(\theta)$. Then, for a given $M$, $C_M$ is known since $\Theta$ is known under the pairwise restrictions, and thus one can compute $\mathcal{PLR}_n(M)$. Finally, enumerate $M$ on a fine grid and all the $M$'s with $\mathcal{PLR}_n(M)$ smaller than $\chi^2_{1, \alpha}$ will be in $CI_{\text{PL}}$. This is implemented in our R package \if1\blind
{
  [Blinded for review].
} \else{
  \pkg{crc.partialid}. 
}\fi
The existing proofs of the asymptotic distribution of profile likelihood ratio statistics relies on point identification as an essential condition (e.g. \citep{rao2017advanced}), so we do not have theoretical guarantees of the asymptotic performance of the proposed $CI_{\text{PL}}$ as $M$ is possibly partially identified. However, simulations demonstrate good performance of $CI_{\text{PL}}$. % with much less computation needed than for $CI_{\text{TIB}}$.  

%\comment{At the same time, related theoretical work has appeared recently using $\chi^2$ random variables to asymptotically stochastically dominate (profile) likelihood ratio statistics for partially identified parameters. For example,  \citet{chen2018monte} studied a class of partially-identified models with one-dimensional sub-vectors of interest using profile likelihood ratio statistics, though their regularity conditions can be difficult to verify. Also, \citet{zhang2009likelihood} studied asymptotic distribution of likelihood ratio statistics for specific applications including non-response and censoring under partial identification. We leave the the theoretical study of the asymptotic properties of the profile likelihood confidence interval for partially identified parameters as an open question for future work.}

%%%%%%%%%%%%%%%%%%%%%%%%%%%

\section{Simulations}
\subsection{Hierarchical log-linear models}

We show the simulation results for all hierarchical log-linear models with two or three samples in Figure \ref{fig:simu_inde}.  The results reflect the restrictive nature of hierarchical log-linear models. %With $k$ CRC samples, a hierarchy of $2^k$ models with different independence assumptions are available. 
Since each model is point-identified, it aims to only cover its hypothesized identified population size. As a result, only certain points of the whole space of the population size is likely to be covered at the nominal level. In fact, as the results suggest, only a proportion of such models can actually attain nominal coverage level, due to wrong restrictions on sample dependence. All the models except the saturated model ``[12,13,23]'' are very sensitive to model misspecification, indicated by the quick decay of coverage probabilities from their peaks. This illustrates why the common practice of model selection using criteria like BIC may not be helpful. In our simulations, the coverage probability of the model with lowest BIC is near $80\%$ at its peak, much lower than the nominal $95\%$.

\begin{figure}[h]%[H]
  \centering
  \includegraphics[width=\textwidth]{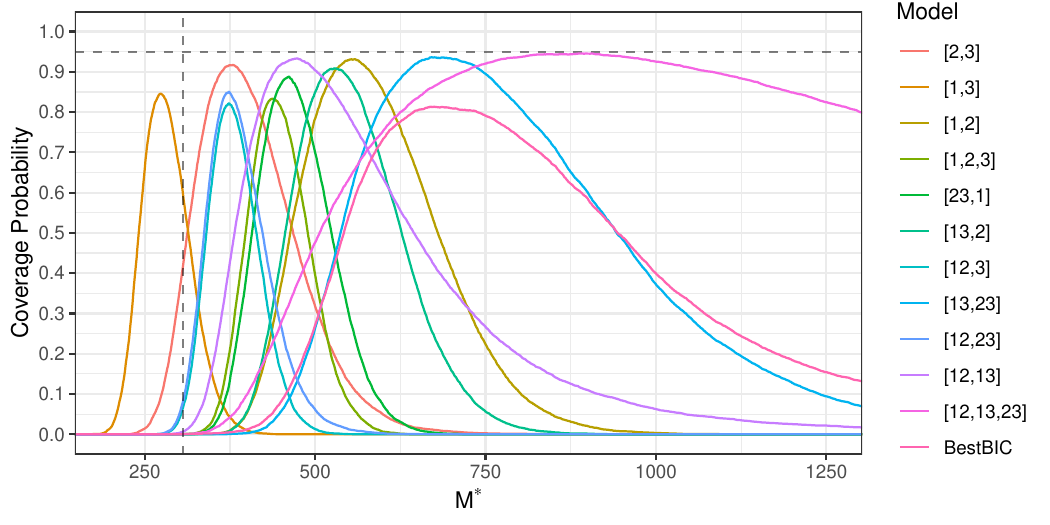}
  \caption{Coverage probability of nominal 95\% confidence intervals for the true population size $M^\ast$ from Poisson hierarchical log-linear models. The confidence intervals were computed by R package``Rcapture" across 6000 simulated tables. The horizontal dashed line represents the nominal coverage probability $95\%$. The vertical dashed line is the minimum possible population size given the true means, $\sum_{i = 1}^c m_i^{\ast} = 306$. The method ``BestBIC" is the model with the lowest BIC across all models using 3 samples.  }
  \label{fig:simu_inde}
\end{figure}

\subsection{Results for $k = 5, 10$}

For $k = 5, 10$, we use the log-linear model $\log m_i = \blambda'\x_i$ to generate $m_i's$, and then generate contingency tables using Poisson distribution with mean $\m$. $\blambda$ is generated by random number generators once for each $k$ and then fixed:\\
\begin{equation}
  \begin{aligned}
    &\lambda_0 = \log(m_0), m_0 \equiv 10^4,\\
    &\lambda_i \sim \text{Uniform}(-3,-1.5), 1 \le i \le k,\\
    &\lambda_i \sim \text{Uniform}(-0.5,0.5),  k < i \le c.
  \end{aligned}
\end{equation}

We impose two sets of flexible pairwise restrictions for each case. Let $\omega$ denote the number of restrictions. For $k = 5$, we study $\omega = 3, 6$; for $k = 10$, we study $\omega = 4, 8$. We detail the conditions below using compact notations in the Method section:\\
(a) $k = 5, \omega = 3$ (corresponding identification region $(13737, 29349)$)
\begin{equation*}
  \begin{aligned}
    r_{5,3} &= (1, 1, 1)\\
    t_{5,3} &= (2, 3, 4)\\
    \eta_{5,3} &= (0.1, 1, 0.3)\\
    \xi_{5,3} &=  (10, 10, 3)
  \end{aligned}
\end{equation*}

(b) $k = 5, \omega = 6$ (corresponding identification region $(23258, 25789)$)
\begin{equation*}
  \begin{aligned}
    r_{5,6} &= (1, 1, 1, 2, 2, 3)\\
    t_{5,6} &= (2, 3, 4,  4, 5, 4)\\
    \eta_{5,6} &= (0.1, 1, 0.3,  0, 1, 0.8)\\
    \xi_{5,6} &=  (10, 10, 3,     1,   +\infty,   5)
  \end{aligned}
\end{equation*}

(c) $k = 10, \omega = 4$ (corresponding identification region $(29739, 35100)$)
\begin{equation*}
  \begin{aligned}
    r_{10,4} &= (1, 2, 3, 2)\\
    t_{10,4} &= (2, 3, 4, 7)\\
    \eta_{10,4} &= (1, 0.1, 1, 0.5)\\
    \xi_{10,4} &=  (+\infty,10,10,1)
  \end{aligned}
\end{equation*}

(d) $k = 10, \omega = 8$ (corresponding identification region $(29739, 30598)$)
\begin{equation*}
  \begin{aligned}
    r_{10,8} &= (1, 2, 3, 2, 5, 6, 8, 9)\\
    t_{10,8} &= (2, 3, 4, 7, 6, 7, 9, 10)\\
    \eta_{10,8} &= (1, 0.1, 1, 0.5, 0.2, 1, 0.5, 0.8)\\
    \xi_{10,8} &=  (+\infty,10,10,  1,  5,  3,   2,  1.2)
  \end{aligned}
\end{equation*}

We generate contingency tables for $1000$ times for $k = 5,10$, and for each table, we apply TIB and PL to the observable data under the above restrictions, using our open source R package. The coverage probabilities for each scenario are shown in Figures \ref{fig:simu_r2_5_3}, \ref{fig:simu_r2_5_6}, \ref{fig:simu_r2_10_4}, \ref{fig:simu_r2_10_8} below. These results indicate that TIB confidence intervals are always valid, tending to be conservative generally, while PL confidence intervals tend to be more anti-conservative when $k$ or $\omega$ becomes larger. Computationally, the time needed to compute the TIB interval is insensitive to $k$, and grows marginally in proportion to $\omega$, usually within 3 minutes on a MacBook Pro with a 3.1 GHz Dual-Core Intel Core i5 processor. In contrast, computation time for PL grows quickly with $k$, making it best suited to CRC studies with small $k$.

\begin{figure}[h]%[H]
  \centering
  \includegraphics[width=\textwidth]{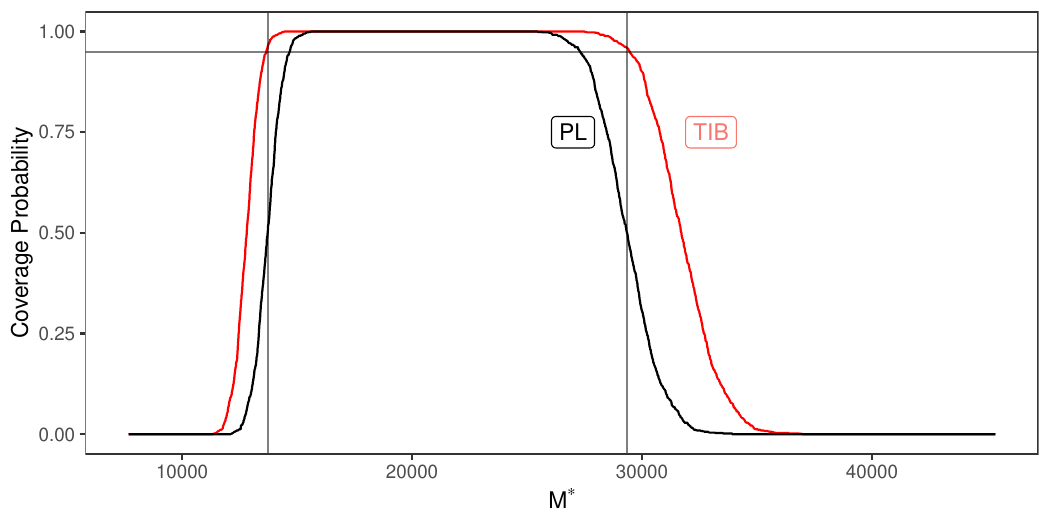}
  \caption{Coverage probability of nominal 95\% confidence intervals $CI_{\text{TIB}}$ (black line) and $CI_{\text{PL}}$ (red line) for the population size $M^*$ under setting (a) ($k = 5, \omega = 3$) across 1000 simulated tables. The horizontal line represents the nominal coverage probability $95\%$. The two vertical lines represent the lower and upper endpoints of identification sets.} 
  \label{fig:simu_r2_5_3}
\end{figure}

\begin{figure}[h]%[H]
  \centering
  \includegraphics[width=\textwidth]{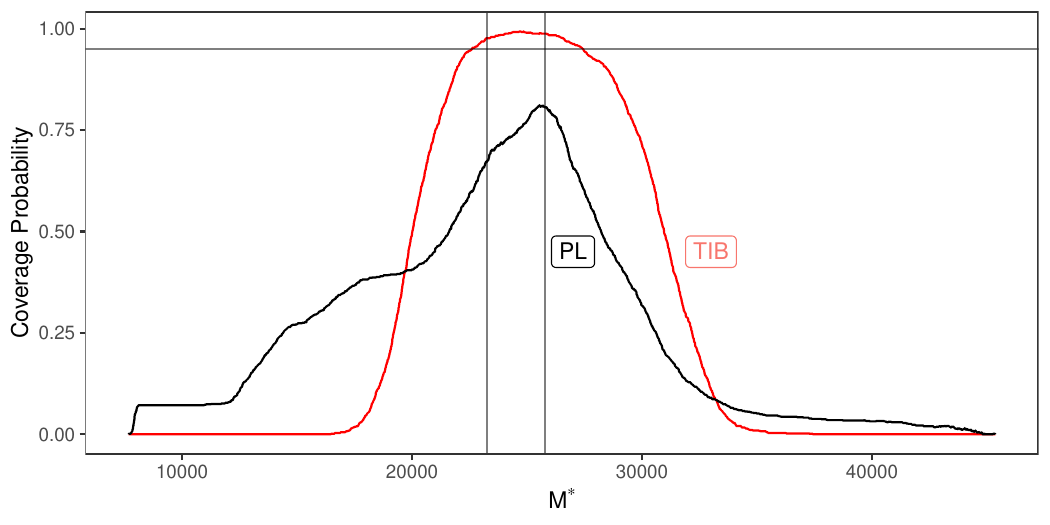}
  \caption{Coverage probability of nominal 95\% confidence intervals $CI_{\text{TIB}}$ (black line) and $CI_{\text{PL}}$ (red line) for the population size $M^*$ under setting (b) ($k = 5, \omega = 6$) across 1000 simulated tables. The horizontal line represents the nominal coverage probability $95\%$. The two vertical lines represent the lower and upper endpoints of identification sets.} 
  \label{fig:simu_r2_5_6}
\end{figure}

\begin{figure}[h]%[H]
  \centering
  \includegraphics[width=\textwidth]{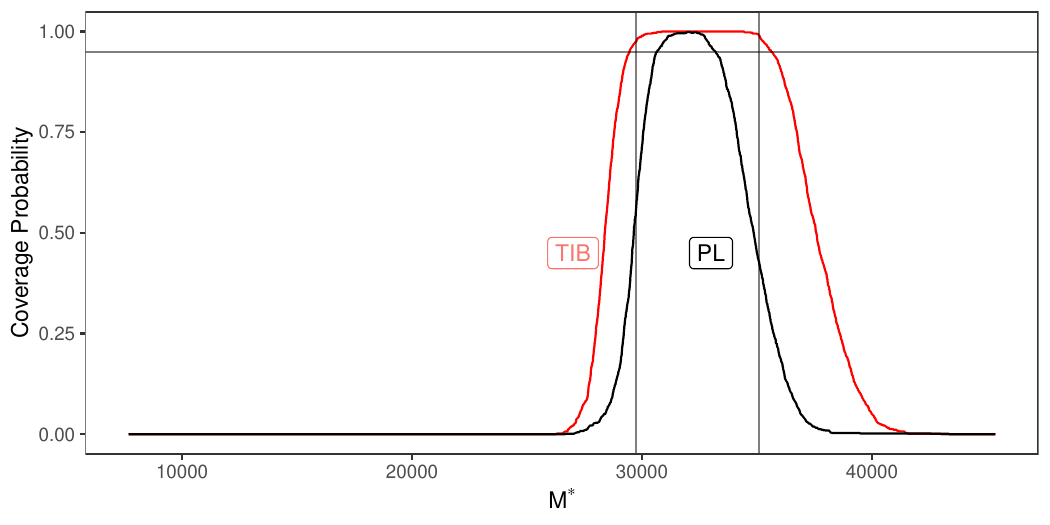}
  \caption{Coverage probability of nominal 95\% confidence intervals $CI_{\text{TIB}}$ (black line) and $CI_{\text{PL}}$ (red line) for the population size $M^*$ under setting (c) ($k = 10, \omega = 4$) across 1000 simulated tables. The horizontal line represents the nominal coverage probability $95\%$. The two vertical lines represent the lower and upper endpoints of identification sets.} 
  \label{fig:simu_r2_10_4}
\end{figure}

\begin{figure}[h]%[H]
  \centering
  \includegraphics[width=\textwidth]{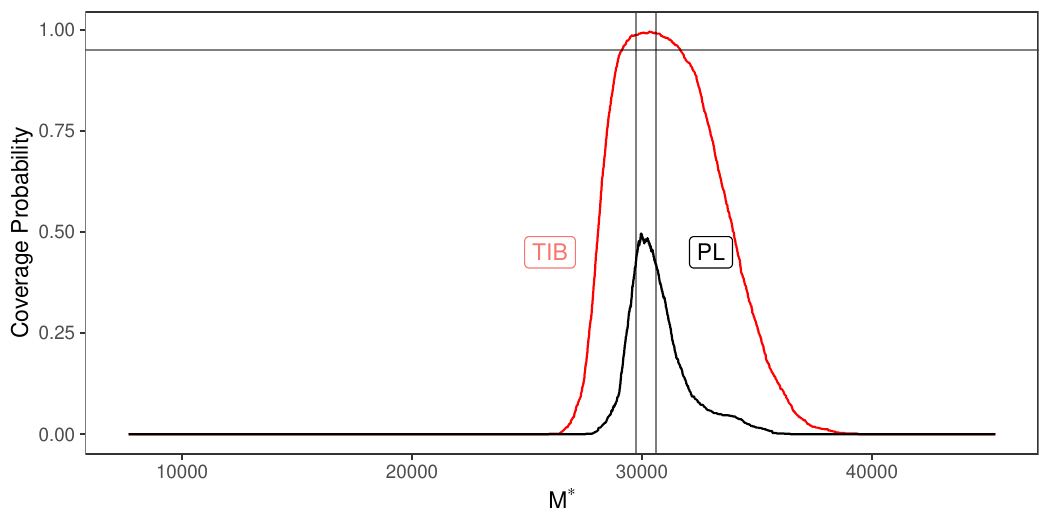}
  \caption{Coverage probability of nominal 95\% confidence intervals $CI_{\text{TIB}}$ (black line) and $CI_{\text{PL}}$ (red line) for the population size $M^*$ under setting (d) ($k = 10, \omega = 8$) across 1000 simulated tables. The horizontal line represents the nominal coverage probability $95\%$. The two vertical lines represent the lower and upper endpoints of identification sets.} 
  \label{fig:simu_r2_10_8}
\end{figure}

\subsection{Lengths of CIs when $k = 3$}

When the identification assumptions are equally credible, the length of a confidence interval indicates how informative and efficient an inferential procedure is. Thus, for correctly specified models, we report the coverage probability of the confidence interval $CI$ for the identification set $M_I(P)$ (i.e. $P \left(  M_I(P) \subseteq CI \right)$), as well as the average length of confidence intervals, as a partial summary of a method's performance in Figure \ref{fig:simu_cover_set}. 

\begin{figure}%[H]
  \centering
  \includegraphics{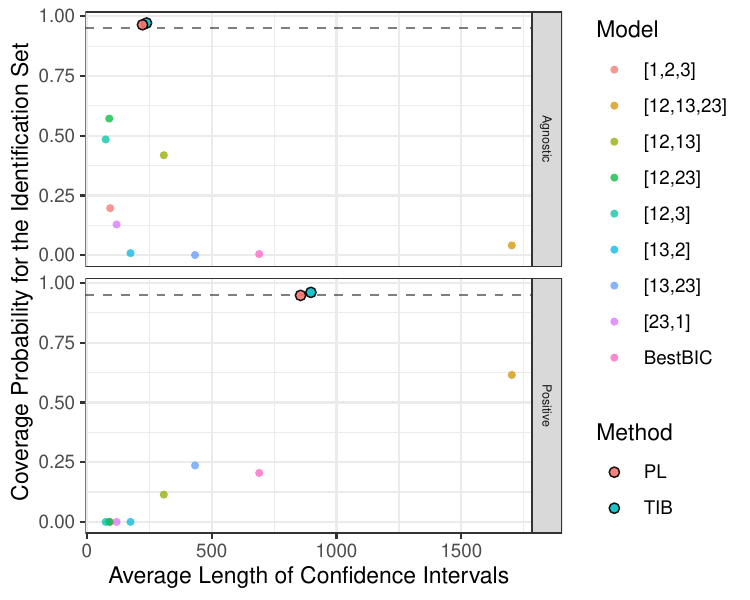}
  \caption{The coverage probability and average length of nominal 95\% confidence intervals for the identification set across 6000 simulated tables when the model is correctly specified. The horizontal line is $95\%$. \emph{Agnostic} pairwise restriction: $\eta = 1/3, \xi = 3$; \emph{Positive} pairwise restriction: $\eta = 1, \xi = 5$. \emph{TIB}: Test inversion bootstrap CIs; \emph{PL}: profile likelihood CIs. }. 
  \label{fig:simu_cover_set}
\end{figure}

It can be seen that the average length of our confidence intervals can be comparable to log-linear models with strong hierarchical assumption, with much higher coverage probability at the same time.

\subsection{Type B violations of Assumption \ref{assum:correctness} when $k = 3$}

In the main text, we have shown results for Type A violations of Assumption \ref{assum:correctness}. In Figure \ref{fig:simu_part_inc}, we show results for Type B with pairwise restrictions $\eta = 1, \xi = 3$ (Positive) and $\eta = 1/2, \xi = 2$ (Agnostic).

\begin{figure}[h]%[H]
  \centering
  \includegraphics[width=\textwidth]{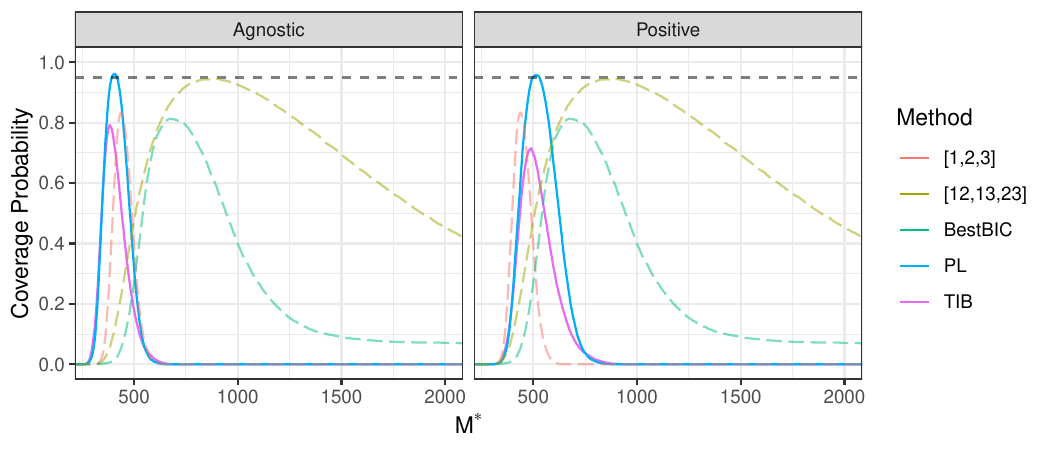}
  \caption{Coverage probability of nominal 95\% confidence intervals $CI_{\text{TIB}}$ and $CI_{\text{PL}}$ for the population size $M^*$ under Type B violations across 6000 simulated tables. Three models with independence assumptions are also presented as a reference: the independence model [1,2,3], the saturated model [12,13,23], and BestBIC model. The horizontal line represents the nominal coverage probability $95\%$. \emph{Agnostic} pairwise restriction: $\eta = 1/2, \xi = 2$; \emph{Positive} pairwise restriction: $\eta = 1, \xi = 3$.} 
  \label{fig:simu_part_inc}
\end{figure}

\paragraph{Remark 1:} If $\m^\ast$ is known and $M^\ast$ is unspecified, one can still compute lower bounds of $\OR_{rt}$ for $r$th and $t$th samples. Since $\OR_{rt} = \frac{m_{11}(r,t)(m_{00}(r,t) + m_0)}{m_{10}(r,t)m_{01}(r,t)} > \frac{m_{11}(r,t)m_{00}(r,t) }{m_{10}(r,t)m_{01}(r,t)}$, we know $\OR_{1,2} > 0.078, \OR_{1,3} > 1.5, \OR_{2,3} > 0.56$ in our simulations. Also, in the simulations, we do not truncate confidence intervals at the count of total observed units. 

\paragraph{Remark 2:} It may occur that when computing $CI_{\text{TIB}}$, for all $M$, the null hypothesis $H_M$ is rejected, and thus $CI_{\text{TIB}} = \emptyset$, which cannot cover the identification set. In our simulations, the probability of this occurrence in correctly specified models is less that $0.02\%$, and is less than $6\%$ for incorrectly specified models. All the calculations of coverage probabilities reported include this situation in the denominator.

%\doublespacing

%%%%%%%%%%%%%%%%%%%%%%%%%%%%%%%%%%%%%%%%

\newpage
\bibliographystyle{unsrtnat}
\bibliography{crc_outline}